\documentclass[12pt,draftcls,onecolumn]{IEEEtran}
\usepackage{amssymb,bbm,amsmath,color,psfrag,subfigure}
\usepackage[draft]{graphicx}
\usepackage{xspace}
\usepackage{algorithm,algorithmic} 
\usepackage{wrapfig}
\usepackage{tikz}
\usepackage{amssymb,amsmath,color,bm,bbm,mathrsfs,amscd}
\usepackage{dsfont}
\usepackage{tikz}

\def\qed{\hfill \vrule height 7pt width 7pt depth 0pt\medskip}
\def\beq{\begin{equation}}
\def\eeq{\end{equation}}
\def\proof{\noindent{\bf Proof}\ \ }

\newcommand{\ds}{\displaystyle}

\newcommand{\ba}{\begin{array}}
\newcommand{\ea}{\end{array}}

\renewcommand{\l}{\left}\renewcommand{\r}{\right}
\newcommand{\be}{\begin{equation}}
\newcommand{\ee}{\end{equation}}
\newcommand{\eps}{\varepsilon}

\newcommand{\1}{\mathbbm{1}}

\newcommand{\R}{\mathbb{R}}

\newcommand{\de}{\mathrm{d}}

\newcommand{\se}{\text{ if }}

\DeclareMathOperator{\cl}{cl}

\DeclareMathOperator*{\argmin}{argmin}

\newcommand{\real}{{\mathbb{R}}}
\newcommand{\realnoneg}{\R_+}

\newcommand{\subscr}[2]{{#1}_{\textup{#2}}}
\newcommand{\supscr}[2]{{#1}^{\textup{#2}}}
\newcommand{\union}{\cup}

\newcommand{\until}[1]{\{1,\dots,#1\}}

\newcommand{\mc}{\mathcal}

\newcommand{\densitymax}{\supscr{\rho}{max}}

\newcommand{\flowmax}{\supscr{f}{max}}
\newcommand{\flowmaxe}{\subscr{f}{e}^\textup{max}}
\newcommand{\newflowmaxe}{\subscr{\tilde{f}}{e}^\textup{max}}

\newcommand{\minimize}{\text{minimize}}
\newcommand{\subj}{\text{subj. to}}

\newcommand{\flow}{f}
\newcommand{\floweq}{{\flow}^{\text{*}}}
\newcommand{\floweqwardrop}{{\flow}^{\text{W}}}
\newcommand{\rhoeq}{{\rho}^{\text{*}}}

\newcommand{\floweqe}{\supscr{\flow}{*}_e}

\newcommand{\densitymaxe}{\supscr{\rho}{max}_e}

\newcommand{\priceofanarchy}{P}

\newcommand{\graph}{\mc T}

\newcommand{\edgeset}{\mc E}

\newcommand{\paths}{\mc P}

\newcommand{\path}{p}
\newcommand{\delayfunc}{T}
\newcommand{\flowsymb}{f}

\newcommand{\toll}{\Upsilon}

\newcommand{\onebf}{\mathbf{1}}
\newcommand{\zerobf}{\mathbf{0}}

\newcommand{\ilogitconst}{\eta}

\newtheorem{assumption}{Assumption}

\newtheorem{theorem}{Theorem}
\newtheorem{proposition}{Proposition}
\newtheorem{corollary}{Corollary}
\newtheorem{definition}{Definition}
\newtheorem{lemma}{Lemma}
\newtheorem{remark}{Remark}

\newtheorem{example}{Example}

\newcommand\oprocendsymbol{\hbox{$\square$}}
\newcommand\oprocend{\relax\ifmmode\else\unskip\hfill\fi\oprocendsymbol}

\begin{document}

\title{Robust Distributed Routing in Dynamical Flow Networks -- Part II: Strong Resilience, Equilibrium Selection and Cascaded Failures}

\author{Giacomo Como\thanks{G.~Como, K.~Savla, M.A.~Dahleh and E.~Frazzoli are with the Laboratory for Information and Decision Systems at the Massachusetts Institute of Technology. \texttt{\{giacomo,ksavla,dahleh,frazzoli\}@mit.edu.}} \quad Ketan Savla \quad Daron Acemoglu\thanks{D.~Acemoglu is with the Department of Economics at the Massachusetts Institute of Technology. \texttt{daron@mit.edu}.} \quad Munther A. Dahleh \quad Emilio Frazzoli \thanks{This work was supported in part by NSF EFRI-ARES grant number 0735956. Any opinions, findings, and conclusions or recommendations expressed in this publication are those of the authors and do not necessarily reflect the views of the supporting organizations. G. Como and K. Savla thank Prof. Devavrat Shah for helpful discussions.
A preliminary version of this paper appeared in part as \cite{Como.Savla.ea:MTNS10}.
}\\\vspace{-40pt}} 

\maketitle
\begin{abstract}
Strong resilience properties of dynamical flow networks are analyzed for distributed routing policies. The latter are characterized by the property that the way the inflow at a non-destination node gets split among its outgoing links is allowed to depend only on local information about the current particle densities on the outgoing links. The strong resilience of the network is defined as the infimum sum of link-wise flow capacity reductions under which the network cannot maintain the asymptotic total inflow to the destination node to be equal to the inflow at the origin. A class of distributed routing policies that are locally responsive to local information is shown to yield the maximum possible strong resilience under such local information constraints for an acyclic dynamical flow network with a single origin-destination pair. The maximal strong resilience achievable is shown to be equal to the minimum node residual capacity of the network. The latter depends on the limit flow of the unperturbed network and is defined as the minimum, among all the non-destination nodes, of the sum, over all the links outgoing from the node, of the differences between the maximum flow capacity and the limit flow of the unperturbed network. We propose a simple convex optimization problem to solve for equilibrium limit flows of the unperturbed network that minimize average delay subject to strong resilience guarantees, and discuss the use of tolls to induce such an equilibrium limit flow in transportation networks. Finally, we present illustrative simulations to discuss the connection between cascaded failures and the resilience properties of the network.

\end{abstract}


\section{Introduction}
Robustness of routing policies for flow networks is a central problem which is gaining increased attention with a growing awareness to safeguard critical infrastructure networks against natural and man-induced disruptions. Information constraints limit the efficiency and resilience of such routing policies, and the possibility of cascaded failures through the network adds serious challenges to this problem. The difficulty is further magnified by the presence of dynamical effects \cite{Simonsen.Buzna.ea:08}. 

This paper considers the framework of \emph{dynamical flow networks} introduced in our companion paper~\cite{PartI}, where the network is modeled by a system of ordinary differential equations derived from mass conservation laws on directed acyclic graphs with a single origin-destination pair and a constant inflow at the origin. The rate of change of the particle density on each link of the network equals the difference between the \emph{inflow} and the \emph{outflow} on that link. The latter is modeled to depend on the current particle density on that link through a \emph{flow function}. 
We focus on \emph{distributed routing policies} whereby the proportion of incoming flow routed to the outgoing links of a node is allowed to depend only on \emph{local information}, consisting of the current particle densities on the outgoing links of the same node. We call the dynamical flow network \emph{fully transferring} if the outflow at the destination node asymptotically approaches the inflow at the origin node.
Our primary objective in this paper is to analyze the robustness of distributed routing policies in terms of the network's \emph{strong resilience}, which is defined as the infimum sum of link-wise magnitude of disturbances making the perturbed dynamical flow network not fully transferring. 

We prove that the maximum possible strong resilience is yielded by a class of \emph{locally responsive} distributed routing policies, introduced in the companion paper~\cite{PartI}. Such policies are characterized by the property that the portion of its inflow that a node routes towards an outgoing link does not decrease as the particle density on any other outgoing link increases. We show that the strong resilience of a dynamical flow network with such locally responsive distributed routing policies equals the \emph{minimum node residual capacity}. The latter is defined as the minimum, among all the non-destination nodes, of the sum of the difference between the maximum flow capacity and the limit flow of the unperturbed network, on all the links outgoing from the node. 
Using idea from \cite{Como.Savla.ea:Wardrop-arxiv}, one can show that, when the information constraints on the routing policies are relaxed, i.e., the routing policies can access information about the particle densities over the whole network, then the strong resilience of the network is equal to the network residual capacity. The latter is defined as the difference between the min-cut capacity of the network and rate of arrival at the origin node. Since the minimum node residual capacity is in general less than the network residual capacity, this shows that the information constraints on the routing policies reduce the strong resilience of the network. Moreover, the minimum residual capacity depends on the limit flow of the unperturbed network. This is in stark contrast to our result on weak resilience in \cite{PartI}, where we showed that the weak resilience is unaffected by local information constraints on the routing policies and is independent of the limit flow of the unperturbed network.
We also propose a simple convex optimization problem to solve for equilibrium limit flows of the unperturbed network that minimize average delay subject to strong resilience guarantees, and discuss the use of tolls to induce such an equilibrium limit flow in transportation networks. These results are derived under the condition when the link-wise flow functions are strictly increasing and the links have unbounded capacity for flow densities. 
We present illustrative simulations discussing cascaded failures that arise when the links have finite capacities on flows as well as densities. It is noteworthy that, we not only describe cascaded failures within a dynamical flow network framework and formalize their effect by establishing the connection to our notions of network resilience, but also highlight the role of distributed routing policies in averting such failures.

Stability analysis of network flow control policies under various routing policies is carried out in \cite{Sengoku.Shinoda.ea:88,Tassiulas.Ephremides:92,Low.Paganini.ea:02}. A detailed comparison between the settings of these papers and our dynamical flow network setting is included in the companion paper~\cite{PartI}.
This paper also studies the connection between the robustness properties of the network and its equilibrium flow. 
The role of equilibrium in the efficiency of a system, especially in economic settings involving multiple agents, has attracted a lot of attention, e.g., see \cite{Dubey:86}. One of the most celebrated notions to measure the inefficiency of an equilibrium is the \emph{price of anarchy}~\cite{Roughgarden:05}. In a transportation setting, the price of anarchy of a given network state quantifies the extent to which the average delay faced by a driver at that state exceeds the least possible average delay over all network states. In this paper, we propose a robustness-based metric for measuring inefficiency of equilibrium states of dynamical flow networks. 
Finally, the study of cascaded failure for complex networks has attracted a great deal of attention recently, e.g., see \cite{Motter.Lai:02,Crucitti.Latora.ea:04} where the authors propose various models to explain this phenomenon.

The contributions of this paper are as follows: (i) we formulate the notion of strong resilience of a dynamical flow network, and show that the class of locally responsive routing policies yield the maximum strong resilience under local information constraint; (ii) we formulate a simple convex optimization problem to solve for the most robust equilibrium flow, and discuss the use of tolls in implementing such an equilibrium in transportation networks; and (iii) we present illustrative simulations to discuss cascaded failures in dynamical flow networks and their effect on network resilience.

The rest of the paper is organized as follows. In Section~\ref{sec:model}, we briefly summarize the dynamical flow network framework and the postulate the notion of strong resilience. In Section~\ref{sec:perturbations}, we state the main result on the strong resilience, and provide discussions on the results. Section~\ref{sec:equilibrium-selection} discusses the problem of selection of the most strongly resilient equilibrium flow of the network and the use of tolls to induce such an equilibrium in transportation networks. In Section~\ref{sec:sim}, we report illustrative numerical simulation results, discussing the effect of cascading failures on the resilience of the network. We conclude in Section~\ref{sec:conclusion} with remarks on future research directions and state proofs of the main results in the appendices A and B.

Before proceeding, we define some preliminary notation to be used throughout the paper. Let $\real$ be the set of real numbers, $\realnoneg:=\{x\in\R:\,x\ge0\}$ be the set of nonnegative real numbers. Let $\mc A$ and $\mc B$ be finite sets. Then, $|\mc A|$ will denote the cardinality of $\mc A$, $\R^{\mc A}$ (respectively, $\R_+^{\mc A}$) the  space of real-valued (nonnegative-real-valued) vectors whose components are indexed  by elements of $\mc A$, and $\R^{\mc A\times\mc B}$ the space of matrices whose real entries indexed  by pairs of elements in $\mc A\times\mc B$. The transpose of a matrix $M \in \real^{\mc A \times\mc B}$, will be denoted by $M^T \in\R^{\mc B\times\mc A}$, while $\onebf$ the all-one vector, whose size will be clear from the context. Let $\cl(\mc X)$ be the closure of a set $\mc X\subseteq\R^{\mc A}$. A directed multigraph is the pair $(\mc V,\mc E)$ of a finite set $\mc V$ of nodes, and of a multiset $\mc E$ of links consisting of ordered pairs of nodes (i.e., we allow for parallel links). Given a a multigraph $(\mc V,\mc E)$, for every node $v\in\mc V$, we shall denote by $\mc E^+_v\subseteq\mc E$, and $\mc E^-_v\subseteq\mc E$, the set of its outgoing and incoming links, respectively. Moreover, we shall use the shorthand notation $\mc R_v:=\R_+^{\mc E^+_v}$ for the set of nonnegative-real-valued vectors whose entries are indexed  by elements of $\mc E^+_v$, $\mc S_v:=\{p\in \mc R_v:\,\sum_{e \in \mc E_v^+} p_e=1\}$ for the simplex of probability vectors over $\mc E^+_v$, and $\mc R:=\R_+^{\mc E}$ for the set of nonnegative-real-valued vectors whose entries are indexed  by the links in $\mc E$.

\section{Dynamical flow networks}
\label{sec:model}
The notion of dynamical flow network was introduced in the companion paper~\cite{PartI}. In order to render the present paper self-contained, we introduce here the concepts and notation which are most relevant. We start with the following definition of a flow network.\medskip

\begin{definition}[Flow network]\label{def:flownetwork}
A \emph{flow network} $\mc N=(\mc T,\mu)$ is the pair of a \emph{topology}, described by a finite directed multigraph $\mc T=(\mc V,\mc E)$, where $\mc V$ is the node set and $\mc E$ is the link multiset, and a family of \emph{flow functions} $\mu:=\{\mu_e:\R_+\to\R_+\}_{e\in\mc E}$ describing the functional dependence $f_e=\mu_e(\rho_e)$ of the flow on the density of particles on every link $e\in\mc E$. 
The  \emph{flow capacity} of a link $e\in\mc E$ is  
\be \flowmax_e:=\sup_{\rho_e \geq 0} \mu_e(\rho_e)\,. \ee
\end{definition}\medskip

We shall use the notation $\mc F_v:=\times_{e\in\mc E^+_v}[0,\flowmax_e)$ for the set of admissible flow vectors on outgoing links from node $v$, and $\mc F:=\times_{e\in\mc E}[0,\flowmax_e)$ for the set of admissible flow vectors for the network. We shall write $f:=\{f_e:\,e\in\mc E\}\in\mc F$, and $\rho:=\{\rho_e:\,e\in\mc E\}\in\mc R$, for the vectors of flows and of densities, respectively, on the different links. The notation $f^v:=\{f_e:\,e\in\mc E^+_v\}\in\mc F_v$, and $\rho^v:=\{\rho_e:\,e\in\mc E^+_v\}\in\mc R_v$ will stand for the vectors of flows and densities, respectively, on the outgoing links of a node $v$. We shall compactly denote by $f=\mu(\rho)$ and $f^v=\mu^v(\rho^v)$ the functional relationships between density and flow vectors.

Throughout this paper, we shall restrict ourselves to flow networks satisfying the following assumptions. \medskip
\begin{assumption}\label{ass:acyclicity}
The topology $\mc T$ contains no cycles, has a unique origin (i.e., a node $v\in\mc V$ such that $\mc E^-_v$ is empty), and a unique destination (i.e., a node $v\in\mc V$ such that $\mc E^+_v$ is empty). Moreover, there exists a path in $\mc T$ to the destination node from every other node in $\mc V$. 
\end{assumption}\medskip
\begin{assumption}\label{ass:flowfunction}
For every link $e\in\mc E$, the map $\mu_e:\R_+\to\R_+$ is continuously differentiable, strictly increasing, such that $\mu_e(0)=0$, and $f_e^{\max}<+\infty$. 
\end{assumption}\medskip

In particular, Assumption \ref{ass:acyclicity} implies that (see, e.g., \cite{Cormen.Leiserson:90}) one can identify (in a possibly non-unique way) the node set $\mc V$ with the integer set $\{0,1,\ldots,n\}$, where $n:=|\mc V|-1$, in such a way that 
\be\label{vertexordering}\mc E^-_{v}\subseteq\bigcup\nolimits_{0\le u<v}\mc E^+_{u}\,,\qquad\forall  v=0,\ldots,n\,.\ee
In particular, (\ref{vertexordering}) implies that $0$ is the origin node, and $n$ the destination node in the network topology $\mc T$. An \emph{origin-destination cut}  (see, e.g., \cite{Ahuja.Magnanti.ea:93}) of $\graph$  is a partition of $\mc V$ into $\mc U$ and $\mc V \setminus\mc U$ such that $0 \in \mc U$ and $n \in \mc V \setminus \mc U$. Let $\mc E_{\mc U}^+=\{(u,v)\in\mc E:\,u\in\mc U,v\in\mc V\setminus\mc U\}$ be the set of all the links pointing from some node in $\mc U$ to some node in $\mc V \setminus \mc U$. The \emph{min-cut capacity} of a flow network $\mc N$ is defined as  
\be\label{def:capacity} C (\mc N):=\min_{\mc U} \sum\nolimits_{e \in \mc E_{\mc U}^+} \flowmaxe\,,\ee
where the minimization runs over all the origin-destination cuts of $\mc T$. Throughout this paper, we shall assume a constant inflow $\lambda_0 \ge 0$ at the origin node.
Let us define the set of \emph{admissible equilibrium flows} associated to $\lambda_0$ as  
$$\mc F^*(\lambda_0):=\l\{f^*\in\mc F:\,\sum\nolimits_{e\in\mc E^+_0}f_e^*=\lambda_0,\,\sum\nolimits_{e\in\mc E^+_v}f_e^*=\sum\nolimits_{e\in\mc E^-_v}f_e^*,\,\forall \, 0<v<n\r\}\,.$$
Then, it follows from the max-flow min-cut theorem (see, e.g., \cite{Ahuja.Magnanti.ea:93}), that $\mc F^*(\lambda_0)\ne\emptyset$ whenever $\lambda_0<C(\mc N)$. That is, the min-cut capacity equals the maximum flow that can pass from the origin to the destination while satisfying capacity constraints on the links, and conservation of mass at the intermediate nodes.



We now recall the notion of a distributed routing policy from \cite{PartI}. 

\begin{definition}[Distributed routing policy]\label{def:distributedroutingpolicy}
A \emph{distributed routing policy} for a flow network $\mc N$ is a family of functions $\mc G:=\{G^v:\mc R_v\to\mc S_v \}_{0\le v<n}$ describing the ratio in which the particle flow incoming in each non-destination node $v$ gets split among its outgoing link set $\mc E^+_v$, as a function of the observed current particle density $\rho^v$ on the outgoing links themselves. 
\end{definition}\medskip

The salient feature of Definition \ref{def:distributedroutingpolicy} is that the routing policy $G^v(\rho^v)$ depends only on the \emph{local information} on the particle density $\rho^v$ on the set $\mc E^+_v$ of outgoing links of the non-destination node $v$. 

We now state the definition of a dynamical flow networks and its transfer efficiency.\medskip 

\begin{definition}[Dynamical flow network and its transfer efficiency]\label{def:dynamicalflownetwork}
A \emph{dynamical flow network} associated to a flow network $\mc N$ satisfying Assumption \ref{ass:acyclicity}, a distributed routing policy $\mc G$, and an inflow $\lambda_0\ge0$, is the dynamical system 
\be\label{dynsyst}
\ds\frac{\de}{\de t}\rho_e(t)=\lambda_v(t)G^v_e(\rho^v(t))-f_e(t)\,,\qquad \forall\,0\le v<n\,,\quad\forall\,e\in\mc E^+_v\,,\ee
where
\be\label{felambdadef} f_e(t):=\mu_e(\rho_e(t))\,,\qquad \lambda_v(t):=\l\{\ba{lcl}\lambda_0&\se&v=0\\\sum_{e\in\mc E^-_v}f_e(t)&\se&0<v \leq  n.\ea\r.\ee
\label{def:alphatransferring}
Given some flow vector $f^{\circ}\in\mc F$, the dynamical flow network (\ref{dynsyst}) is said to be  \emph{fully transferring} with respect to $f^{\circ}$ if the solution of (\ref{dynsyst}) with initial condition $\rho(0)=\mu^{-1}(f^{\circ})$ satisfies 
\be\label{alphatransferringdef}\lim_{t\to+\infty}\lambda_n(t) = \lambda_0\,.\ee
\end{definition}
\medskip

Definition \ref{def:alphatransferring} states that a dynamical flow network is fully transferring when the outflow is asymptotically equal to the inflow, i.e., there is no throughput loss asymptotically. 
Observe that a fully transferring dynamical flow network does not necessarily imply that the link-wise flows necessarily converge to an equilibrium, for it might in principle have a persistently oscillatory or more complex behavior. Nevertheless, it will prove useful to introduce the notions of equilibrium and limit flow as follows. 
\medskip

\begin{definition}[Equilibrium and limit flow of a dynamical flow network]\label{def:equilibriumetc}
An \emph{equilibrium flow} for the dynamical flow network (\ref{dynsyst}) is a vector $f^*\in\mc F^*(\lambda_0)$ such that \be\label{equilibriumdef}\lambda_v^* G^v_e(\rho^v)=f_e^*\,,\qquad \forall e\in\mc E^+_v\,,\ \forall 0\le v<n\,,\ee where $\rho^v_e:=\mu_e^{-1}(f_e^*)$, and $\lambda_v^*=\lambda_0$ for $v=0$ and $\lambda_v^*=\sum_{e \in \mc E_v^-} f_e^*$ for $0 < v < n$. \\
A \emph{limit flow} for the dynamical flow network (\ref{dynsyst}) is a vector $f^*\in\cl(\mc F)$ such that the solution of (\ref{dynsyst}) with initial condition $\rho(0)=\mu^{-1}(f^{\circ})$ satisfies
\be\label{limitflow}\lim_{t\to+\infty}f(t)=f^*\,.\ee
The set of all initial flows $f^{\circ}\in\mc F$ such that (\ref{limitflow}) is satisfied will be referred to as the \emph{basin of attraction} of $f^*$, and denoted by $\mc B(f^*)$. 
\end{definition}\medskip

\begin{remark}\label{remark:limitflowequilibrium}
Observe that an equilibrium flow $f^*\in\mc F^*(\lambda_0)$ is always a limit flow, since the solution of the dynamical flow network (\ref{dynsyst}) with initial flow $f^{\circ}=f^*$ stays put for all $t\ge0$, and hence it is trivially convergent to $f^*$. On the other hand, if a limit flow $f^*\in\cl(\mc F)$ satisfies all the capacity constraints with strict inequality, i.e., if $f^*\in\mc F$, then necessarily $f^*\in\mc F^*(\lambda_0)$ is also an equilibrium flow for (\ref{dynsyst}), i.e., it satisfies mass conservation equations at all the non-destination nodes.  In particular, if a dynamical flow network admits an equilibrium flow $f^*$, then it is necessarily fully transferring with respect to $f^*$, as well as with respect to all the initial flows $f^{\circ}\in\mc B(f^*)$. 
  
In contrast, if $f^*\in\cl(\mc F)\setminus\mc F$, i.e., if at least one of the capacity constraints is satisfied with equality, then $f^*$ is not an equilibrium flow for (\ref{dynsyst}). In fact, in this case one has that $\sum\nolimits_{e\in\mc E^+_v}f_e^*\le\lambda_v^*$ with possibly strict inequality for some non-destination node $0\le v<n$. Hence, the dynamical flow network might still be non fully transferring. Finally, observe that a limit flow $f^*\in\cl(\mc F)$ (and, {\it a fortiori}, an equilibrium flow) may not exist for general flow networks $\mc N$, and distributed routing policies $\mc G$. 
 
\end{remark}\medskip

\begin{remark}\label{remark:standardflow}
Standard definitions in the literature are typically limited to static flow networks describing the particle flow at equilibrium via conservation of mass. In fact, they usually consist (see e.g., \cite{Ahuja.Magnanti.ea:93}) in the specification of a topology $\mc T$, a vector of flow capacities $f^{\max}\in\mc R$, and an admissible equilibrium flow vector $f^*\in\mc F^*(\lambda_0)$ for $\lambda_0<C(\mc N)$ (or, often, $f^*\in\cl(\mc F^*(\lambda_0))$ for $\lambda_0\le C(\mc N)$). 

In contrast, in our model we focus on the off-equilibrium particle dynamics on a flow network $\mc N$, induced by a distributed routing policy $\mc G$. Existence of an equilibrium of the dynamical flow network (\ref{dynsyst}) depends on the topology $\mc T$, the structural form of the flow functions $\mu$ and of the distributed routing policy $\mc G$, as well as on the inflow $\lambda_0$. A necessary condition for that is $\lambda_0<C(\mc N)$. In contrast, simple, locally verifiable, sufficient conditions on $\mc G$ for the existence of an equilibrium flow might be hard to find for complex flow networks. However, in some cases, it is reasonable to assume the distributed routing policy $\mc G$ to be the outcome of a slow time-scale evolutionary dynamics with global feedback which can naturally lead to an equilibrium flow $f^*\in\mc F^*(\lambda_0)$. This has been shown, e.g., in our related work \cite{Como.Savla.ea:Wardrop-arxiv} on transportation networks, where the emergence of Wardrop equilibria is proven using tools from singular perturbation theory and evolutionary dynamics. Multiple time-scale dynamics leading to Wardrop equilibria has been studied in \cite{Borkar.Kumar:03} for communication networks.
\end{remark}\medskip
%

While, as discussed in Remark \ref{remark:standardflow}, finding simple, locally verifiable, sufficient conditions on the distributed routing policy $\mc G$ for the existence of an equilibrium flow of the associated dynamical flow network (\ref{dynsyst}) is typically nontrivial, a large class of distributed routing policies was proven to yield existence and uniqueness of a globally attractive limit flow $f^*\in\cl(\mc F)$, as revised below. 

\medskip
\begin{definition}[Locally responsive distributed routing policy]\label{def:myopicpolicy}
A \emph{locally responsive} distributed routing policy for a flow network topology $\mc T=(\mc V,\mc E)$ with node set $\mc V=\{0,1,\ldots,n\}$ is a family of continuously differentiable distributed routing functions $\mc G=\{G^v:\mc R_v\to\mc S_v\}_{v\in\mc V}$ such that, for every non-destination node $0\le v<n$:
\begin{description}
\item[(a)]
$\ds\frac{\partial}{\partial \rho_e}G^v_j(\rho^v)\ge0 \,,\qquad \forall j,e\in\mc E^+_v\,,  j\ne e\,,\rho^v\in\mc R_v\,;$
\item[(b)]
for every nonempty proper subset $\mc J\subsetneq\mc E^+_v$, there exists a continuously differentiable map $G^{\mc J}:\mc R_{\mc J}\to\mc S_{\mc J}$, where $\mc R_{\mc J}:=\R_+^{\mc J}$, and $\mc S_{\mc J}:=\{p\in\mc R_{\mc J}:\,\sum_{j\in\mc J}p_j=1\}$ is the simplex of probability vectors over $\mc J$, such that, for every $\rho^{\mc J}\in\mc R_{\mc J}$, if 
$$\rho^v_e\to+\infty\,,\ \ \forall e\in\mc E^+_v\setminus\mc J\,,\qquad\rho_j\to\rho_j^{\mc J}\,,\ \ \forall j\in\mc J\,,$$ then  $$G^v_e(\rho^v)\to0,\ \ \forall e\in\mc E^+_v\setminus\mc J\,,\qquad G^v_j(\rho)\to G^{\mc J}_j(\rho^{\mc J}),\ \ \forall j\in\mc J\,.$$
\end{description}
\end{definition}

Let us restate the result proven in  \cite[Theorem 1]{PartI}.

\smallskip
\begin{theorem}[Existence of a globally attractive limit flow under locally responsive routing policies]\label{thm:uniquelimitflow}
Let $\mc N$ be a flow network satisfying Assumptions \ref{ass:acyclicity} and \ref{ass:flowfunction}, $\lambda_0\ge0$ a constant inflow, and $\mc G$ a locally responsive distributed routing policy. Then, there exists a unique limit flow $f^*\in\cl(\mc F)$ such that $\mc B(f^*)=\mc F$. Moreover, if $f_e^*=f_e^{\max}$ for some $e\in\mc E^+_v$, and $0\le v<n$, then $f_e^*=f_e^{\max}$, for every $e\in\mc E^+_v$.
\end{theorem}
We shall use the above result in the form of the following corollary, which is an immediate consequence of Theorem \ref{thm:uniquelimitflow} and Remarks \ref{remark:limitflowequilibrium} and \ref{remark:standardflow}. \medskip
\begin{corollary}
Let $\mc N$ be a flow network satisfying Assumptions \ref{ass:acyclicity} and \ref{ass:flowfunction}, $\lambda_0\ge0$ a constant inflow, and $\mc G$ a locally responsive distributed routing policy. If the  limit flow $f^*$ belongs to $\mc F$, then $f^*\in\mc F^*(\lambda_0)$ is a globally attractive equilibrium flow for the dynamical network flow (\ref{dynsyst}), and, consequently, (\ref{dynsyst}) is fully transferring with respect to $f^*$. 
\end{corollary}

\begin{example}[Locally responsive distributed routing policy]
\label{example:routing}
Let $\mc N$ be a flow network satisfying Assumptions \ref{ass:acyclicity} and \ref{ass:flowfunction}, and $0\le\lambda_0<C(\mc N)$ a constant incoming flow. For $\floweq =\mu(\rhoeq)\in\mc F^*(\lambda_0)$, and $\ilogitconst>0$, define the distributed routing policy $\mc G$ by
\begin{equation}
\label{eq:routing-example}
G^v_e(\rho)=\frac{\floweq_e \exp(- \ilogitconst (\rho_e-\rhoeq_e))}{\sum_{j \in \mc E_v^+} \floweq_j \exp(- \ilogitconst (\rho_j-\rhoeq_j))}\,, \qquad \forall e \in \mc E_v^+\,, \quad \forall 0\le v<n\,.
\end{equation}
Then, $\mc G$ can be easily verified to be locally responsive, and $f^*$ to be the globally attractive limit flow of the associated dynamical flow network (\ref{dynsyst}). 
\end{example}\medskip

\section{Strong resilience of dynamical flow networks} \label{sec:perturbations}
In this section, we shall introduce the notion of strong resilience of a dynamical flow network, and show that locally responsive policies are maximally robust among the class of distributed routing policies. We shall also provide an explicit simple characterization of the maximal strong resilience of a dynamical flow network with respect to a given limit flow. 

We shall consider persistent perturbations of the dynamical flow network (\ref{dynsyst}) that reduce the flow functions on the links, as per the following: 
\begin{definition}[Admissible perturbation]\label{def:admissibeperturbation}
An \emph{admissible perturbation} of a flow network $\mc N=(\mc T,\mu)$, satisfying Assumptions \ref{ass:acyclicity} and \ref{ass:flowfunction}, is a flow network $\tilde{\mc N}=({\mc T},\tilde\mu)$, with the same topology $\mc T$, and a family of perturbed flow functions $\tilde\mu:=\{\tilde\mu_e:\R_+\to\R_+\}_{e\in\mc E}$, such that, for every $e\in\mc E$, $\tilde\mu_e$ satisfies Assumption \ref{ass:flowfunction}, as well as
$$\tilde\mu_e(\rho_e)\le\mu_e(\rho_e)\,,\qquad \forall \rho_e\ge0\,.$$
We accordingly let $\newflowmaxe:=\sup\{\tilde{\mu}_e(\rho_e):\rho_e\ge0\}$.
The \emph{magnitude} of an admissible perturbation is defined as 
\be\label{deltadef}\delta:=\sum\nolimits_{e\in\mc E}\delta_e\,,\qquad\delta_e:=\sup\l\{\mu_e(\rho_e)-\tilde\mu_e(\rho_e):\,\rho_e\ge0\r\}\,.\ee 
\end{definition}\medskip

Given a dynamical flow network as in Definition \ref{def:dynamicalflownetwork}, and an admissible perturbation as in Definition \ref{def:admissibeperturbation}, we shall consider the \emph{perturbed dynamical flow network}
\be\label{pertdynsyst}
\ds\frac{\de}{\de t}\tilde\rho_e(t)=\tilde\lambda_v(t)G^v_e(\tilde\rho^v(t))- \tilde f_e(t)\,,\qquad\forall\,0\le v<n\,,\quad\forall\,e\in\mc E^+_v\,,\ee
where
\be\tilde f_e(t):=\tilde\mu_e(\tilde\rho_e(t))\,,\qquad \tilde\lambda_v(t):=\l\{\ba{lcl}\sum_{e\in\mc E^-_v}\tilde f_e(t)&\se&0<v < n\\\lambda_0&\se&v=0\,.\ea\r.
\ee

We are now ready to define the notion of strong resilience of a dynamical flow network as in Definition \ref{def:dynamicalflownetwork} with respect to a limit flow $f^*$.
\medskip
\begin{definition}[Strong resilience of a dynamical flow network]\label{def:stabilitymargins}
Let $\mc N$ be a flow network satisfying Assumptions \ref{ass:acyclicity} and \ref{ass:flowfunction}, $\lambda_0\ge0$ be a constant inflow at the origin, and $\mc G$ a distributed routing policy. Assume that the corresponding dynamical flow network has a limit flow $f^*\in\cl(\mc F)$. The \emph{strong resilience}  $\gamma_{1}(f^*,\mc G)$ is equal to the infimum magnitude of all the admissible perturbations for which the perturbed dynamical flow network (\ref{pertdynsyst}) is not fully transferring with respect to some initial flow $f^{\circ}\in\mc B(f^*)$. 
 \end{definition}\medskip

Note that the notion of strong resilience formalized in Definition~\ref{def:stabilitymargins} is with respect to the worst-case scenario.
Accordingly, one can provide an adversarial interpretation to the perturbations as in \cite{PartI}.
Our first result is an upper bound on the strong resilience of a dynamical flow network driven by an arbitrary distributed routing policy. In order to state such result, for a flow network $\mc N$, and a flow vector $\floweq\in\cl(\mc F)$, define the \emph{minimum node residual capacity} as 
\be\label{gammadef}R(\mc N,\floweq):=\min_{0\le v<n}\l\{\sum\nolimits_{e \in \edgeset_v^+}\left(\flowmaxe - \floweqe \right)\r\}\,.\ee
\smallskip
\begin{theorem}[Upper bound on the strong resilience]\label{lemma:upperbound}
Let $\mc N$ be a flow network satisfying Assumptions \ref{ass:acyclicity} and \ref{ass:flowfunction}, $\lambda_0\ge0$ a constant inflow, and $\mc G$ any distributed routing policy. Assume that the associated dynamical flow network has a limit flow $f^*\in\mc F^*(\lambda_0)$. Then, 
$$\gamma_1(\floweq,\mc G)\le R(\mc N,f^*)\,.$$
\end{theorem}
\begin{proof}See Appendix \ref{sec:proof1}.\end{proof}\medskip

The proof of Theorem \ref{lemma:upperbound} essentially depends only on Assumption \ref{ass:acyclicity} on the acyclicity of the network topology. However, in order to show that the upper bound in Theorem \ref{lemma:upperbound} is tight for locally responsive policies, we have to rely highly on Properties (a) and (b) of Definition \ref{def:myopicpolicy}. The following example illustrates the necessity of these properties.


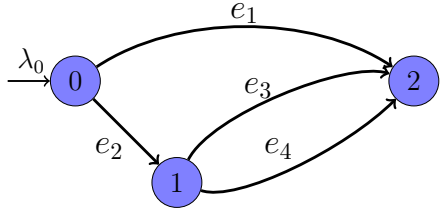
\begin{figure}
\begin{center}
\scalebox{0.9}
{
\begin{tikzpicture}
\path (0,0) node(a) [circle,draw,fill=blue!50!white] {$0$}
		(1.5,-1.5) node (b) [circle,draw,fill=blue!50!white] {$1$}
	(5, 0) node (c) [circle,draw,fill=blue!50!white] {$2$};
\draw[thick,->] (-1,0) -- (a); 
\draw[very thick,->] (a) -- (b);
\draw[very thick,->] (a) .. controls (1.5,1) and (3.5,1) .. (c);
\draw[very thick,->] (b) .. controls (2,-0.5) and (4,0.3) .. (c);
\draw[very thick,->] (b) .. controls (2.4,-1.8) and (4,-1) .. (c);
\draw (-0.65,0.3) node {$\lambda_0$};
\draw (2.5,1) node (d) {\large $e_1$}; 
\draw (0.5,-1) node {\large $e_2$};
\draw (2.7,-0.1) node {\large $e_3$}; 
\draw (3,-1) node {\large $e_4$}; 
\end{tikzpicture}
}
\end{center}
\caption{\label{fig:A5-justification} The network topology used in Example~\ref{ex:non-trivial}.}

\end{figure}

\begin{example}
\label{ex:non-trivial}
\label{example:need-diffusivity}
Consider the topology illustrated in Figure~\ref{fig:A5-justification}, with $\lambda_0=2$, flow functions given by
\begin{equation}
\label{example:flowfunction}
\mu_e(\rho_e)=\flowmax_e \left(1- \exp(-a_e \rho_e) \right)
\end{equation}
with $a_1=a_2=a_3=a_4=1$ and $\flowmax_{e_1}=\flowmax_{e_2}=2$, $\flowmax_{e_3}=\flowmax_{e_4}=0.75$. First consider the case when $G^0_{e_1}(\rho^0)=1-G^0_{e_2}(\rho^0) \equiv 0.75$, and $G^1_{e_3}(\rho^1)=1-G^1_{e_4}(\rho^1) \equiv 0.5$. One can verify that the associated dynamical flow network has a unique equilibrium flow $f^*$ with $\floweq_{e_1}=1.5$, $\floweq_{e_2}=0.5$, and $\floweq_{e_3}=\floweq_{e_3}=0.25$. Now, consider an admissible perturbation such that $\tilde\mu_{e_1}=0.7\mu_{e_1}$ and $\tilde\mu_{e_k}=\mu_{e_k}$ for $k=2,3,4$. The magnitude of such perturbation is $\delta=\delta_{e_1}=0.6$. It is easy to see that in this case $\lim_{t \to \infty} \tilde f_{e_1}(t)=1.4=\tilde f_{e_1}^{\max}$ which is less than $1.5$, which is the flow routed to it. Therefore, $\lim_{t \to \infty} \tilde \lambda_2(t)=1.9<\lambda_0$, and hence the network is not fully transferring. 

Now, consider the same (unperturbed) flow network as before, but with distributed routing policies such that $$G^0_{e_1}(\rho^0)=1-G^0_{e_2}(\rho^0)=2 e^{-0.031 \rho_{e_1}}/(2 e^{-0.031 \rho_{e_1}}+e^{0.7196 \rho_{e_2}})\,,\qquad G^1_{e_3}(\rho^1)=1-G^1_{e_4}(\rho^1) \equiv 0.5\,.$$ One can verify that the associated dynamical flow network again admits the same $f^*$ as before as an equilibrium flow. 
Let us consider the same admissible perturbation as before. One can verify that, for the corresponding perturbed dynamical flow network, $\lim_{t \to \infty} \tilde f_{e_1}(t)=0.4 < \tilde f_{e_1}^{\max}=1.4$ and $\lim_{t \to \infty} \tilde f_{e_2}(t)=1.6 < \tilde f_{e_2}^{\max}=2$. However, with an asymptotic arrival rate of $1.6$ at node $1$, we have that $\lim_{t \to \infty} \tilde f_{e_3}(t)=0.75=\tilde f_{e_3}^{\max}$ and $\lim_{t \to \infty} \tilde f_{e_4}(t)=0.75=\tilde f_{e_4}^{\max}$. Therefore, $\lim_{t \to \infty} \tilde \lambda_2(t)=1.9<\lambda_0$, and hence the network is not fully transferring. 

In both the cases, $R(\mc N,f^*)=1$ and a disturbance of magnitude $0.6$ is enough to ensure that the perturbed dynamical flow network is not fully transferring. However, note that in the second case, unlike the first case, the routing policy at node 
$0$ responds to variations in the local flow densities by sending more flow to link $e_2$, but it is \emph{overly} responsive in the sense that it sends more flow downstream than the cumulative flow capacity of the links outgoing from node $1$. However, by Definition~\ref{def:distributedroutingpolicy}, a distributed routing policy is not allowed any information about any other link other than the current flow densities of its outgoing links. This illustrates one of the challenges in designing distributed routing policies which yield $R(\mc N,f^*)$ as the strong resilience. Observe that this distributed routing policy is not locally responsive, since $G^0$ used in the first case, does not satisfy Property (b) of Definition~\ref{def:myopicpolicy} and, in the second case, it does not satisfy Properties (a) and (b). 
\end{example}\medskip

We now state the main technical result of this paper, showing that, for locally responsive distributed routing function, the strong resilience coincides with the minimal residual node capacity. 

\smallskip
\begin{theorem}[Strong resilience for locally responsive policies]
\label{maintheo}
Let $\mc N$ be a flow network satisfying Assumptions \ref{ass:acyclicity} and \ref{ass:flowfunction}, $\lambda_0\ge0$ a constant inflow, and $\mc G$ a locally responsive distributed routing policy. Let $f^*\in\cl(\mc F)$ be the globally attractive limit flow of the associated dynamical flow network (\ref{dynsyst}). 
Then,  
$$\gamma_1(f^*,\mc G)=R(\mc N,\floweq)\,.$$
\end{theorem}
\proof See Appendix \ref{sec:proof2}.\qed
\medskip

For a given flow network $\mc N$, a constant inflow $\lambda_0$, Theorem \ref{lemma:upperbound} and Theorem \ref{maintheo} imply that, among all distributed routing policies such that the dynamical flow network has a given limit flow $f^*\in\cl(\mc F)$, locally responsive policies (for which such limit flow is unique and globally attractive by Theorem \ref{thm:uniquelimitflow}) have the maximum strong resilience. Moreover, such maximal strong resilience coincides with the minimum node residual capacity $R(\mc N,f^*)$, and hence it depends both on the flow network $\mc N$, and on the limit flow $f^*$ of the unperturbed network. 

\begin{figure}
\begin{center}
\includegraphics[width=9cm,height=6cm]{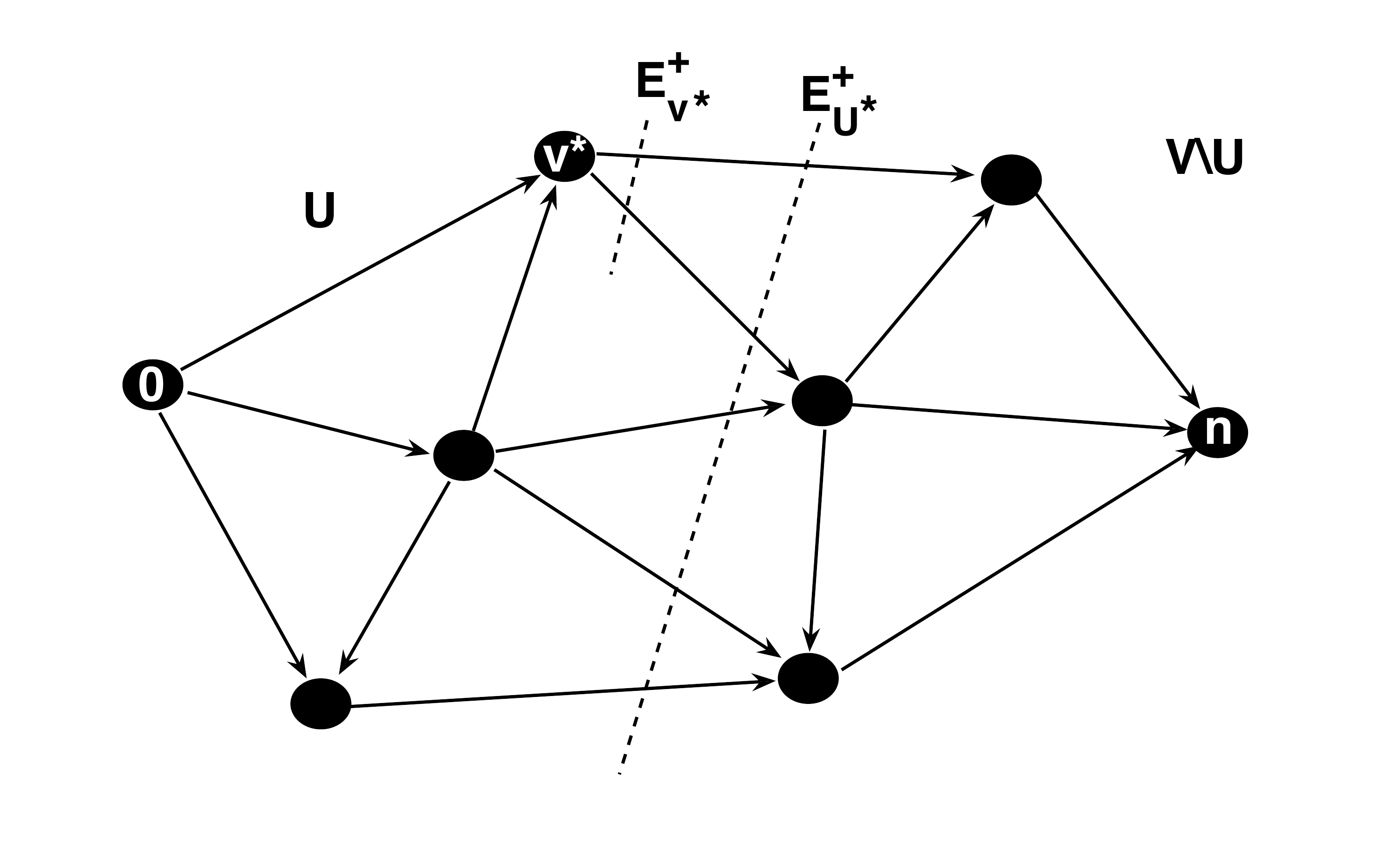} 
\end{center}
\caption{\label{fig:mincutvsnodecut}Comparison between a node-cut and a min-cut of a flow network.}
\end{figure}

A few remarks are in order. First, it is worth comparing the maximum strong resilience of a dynamical flow network with its maximum weak resilience. The latter was studied in \cite{PartI} and was shown (see Definition 6, Proposition 1, and Theorem 2 therein) to be equal to the min-cut capacity of the flow network, $C(\mc N)$. Clearly, the former cannot exceed the latter, as can be also directly verified from the definitions (\ref{gammadef}) and (\ref{def:capacity}): for this, it is sufficient to consider (see Figure \ref{fig:mincutvsnodecut})
$$\mc U^*\in\argmin_{\mc U\text{ origin-destination cut}}\l\{\sum\nolimits_{e\in\mc E^+_{\mc U}}\flowmax_e\r\}\,,\qquad
v^*:=\max\{u\in\mc U^*\}\,,$$ 
and observe that, since $\mc E^+_{v^*}\subseteq\mc E^+_{\mc U^*}$, and $\sum_{e\in\mc E^+_{\mc U^*}}\floweqe=\lambda_0$ by conservation of mass, one has 
$$R(\mc N,f^*)\le \sum_{e\in\mc E^+_{v^*}}(\flowmaxe-\floweqe)\le \sum_{e\in\mc E^+_{\mc U^*}}(\flowmaxe-\floweqe)=\sum_{e\in\mc E^+_{\mc U^*}}\flowmaxe-\lambda_0=C (\mc N)-\lambda_0\,.$$
We provide below two examples to illustrate the difference between the two quantities. 
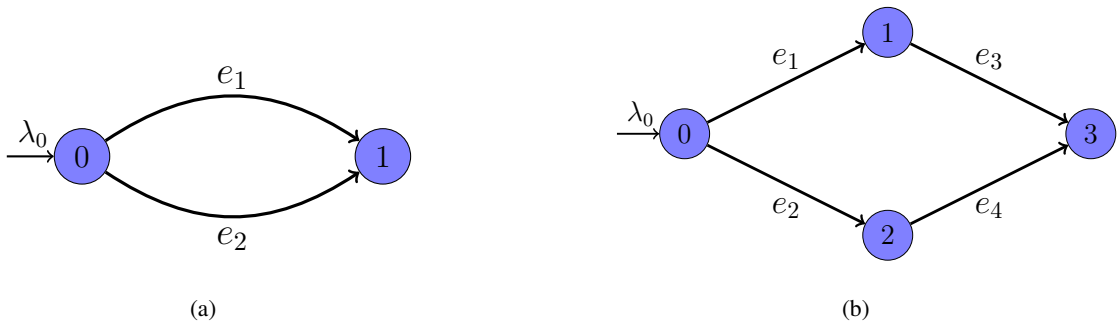
\begin{figure}\begin{center}
\subfigure[]{
\scalebox{1.0}
{\begin{tikzpicture}
\path (0,0) node(a) [circle,draw,fill=blue!50!white] {$0$}
	(4,0) node (b) [circle,draw,fill=blue!50!white] {$1$};
\draw[very thick,->] (a) .. controls (1.5,1) and (2.5,1) .. (b);
\draw [very thick,->] (a) .. controls (1.5,-1) and (2.5,-1) .. (b);
\draw[thick,->] (-1,0) -- (a); 
\draw (-0.65,0.3) node {$\lambda_0$};
\draw (2,1.05) node {\large $e_1$}; 
\draw (2,-1.1) node {\large $e_2$};
\end{tikzpicture}}
}
\hspace{2cm}
\subfigure[]{
\scalebox{0.9}
{
\begin{tikzpicture}

\path (0,0) node(a) [circle,draw,fill=blue!50!white] {$0$}
	(3,1.5) node (b) [circle,draw,fill=blue!50!white] {$1$}
	(3,-1.5) node (c) [circle,draw,fill=blue!50!white] {$2$}
	(6,0) node (d) [circle,draw,fill=blue!50!white] {$3$};
	
\draw[very thick,->] (a) -- (b);
\draw[very thick,->] (a) -- (c);
\draw[very thick,->] (b) -- (d);
\draw[very thick,->] (c) -- (d);

\draw[thick,->] (-1,0) -- (a); 

\draw (-0.65,0.3) node {$\lambda_0$};
\draw (1.5,1.1) node {\large $e_1$}; 
\draw (1.5,-1.1) node {\large $e_2$}; 
\draw (4.5,1.1) node {\large $e_3$}; 
\draw (4.5,-1.1) node {\large $e_4$}; 

\end{tikzpicture}}
}
\end{center}
\caption{(a) A parallel link topology. (b) A topology to illustrate arbitrarily large $ C (\mc N)-R(\mc N,f^*)$.}
\label{fig:large-diff}
\end{figure}
\medskip
\begin{example}
For parallel link topologies, an example of which is illustrated in Figure~\ref{fig:large-diff} (a), one has that $$R(\mc N,f^*)=\sum_{e \in \mc E} \flowmaxe -\lambda_0= C (\mc N)-\lambda_0\,.$$
\end{example}\medskip
\begin{example}
\label{ex:large-strong-weak-gap}
Consider the topology shown in Figure~\ref{fig:large-diff} (b) with $\lambda_0=1$, $\floweq=[\epsilon, 1-\epsilon, \epsilon, 1-\epsilon]$ and $\flowmaxe=[1/\epsilon, 1, 1/\epsilon,1]$ for some $\epsilon \in (0,1)$. In this case, we have that $ C (\mc N)=1+1/\epsilon$ and $R(\mc N,\floweq)=\epsilon$. Therefore, $$ C (\mc N)-R(\mc N,\floweq)=1+1/\epsilon-\epsilon\,,$$ and hence $ C (\mc N)-R(\mc N,\floweq)$ grows unbounded as $\epsilon$ vanishes.
\end{example}\medskip
 
We conclude this section with the following observation. Using arguments along the lines of those employed in \cite{Como.Savla.ea:Wardrop-arxiv}, one can show that $ C (\mc N)-\lambda_0$ provides an upper bound on the strong resilience even if the locality constraint on the information used by the routing policies is removed, i.e., if one allows $G^v$ to depend on the full vector of current densities $\rho$, rather than on the local density vector $\rho^v$ only. Indeed, one might exhibit routing policies which are functions of the global density information $\rho$, for which the strong resilience is exactly $ C (\mc N)-\lambda_0$ using ideas developed in the paper~\cite{Como.Savla.ea:Wardrop-arxiv}. Hence, one may interpret the gap $ C (\mc N)-\lambda_0-R(\mc N,\floweq)$ as the strong resilience loss due to the locality constraint on the information available to the distributed routing policies. One could use Example~\ref{ex:large-strong-weak-gap} to again demonstrate arbitrarily large such loss due to the locality constraint on the information available to the routing policies. 
In fact, it is possible to consider intermediate levels of information available to the routing policies, which interpolate between the one-hop information of our current modeling of the distributed routing policies, and the global information described above. These results on the strong resilience are in stark contrast to our result on weak resilience in \cite{PartI}, where we showed that the weak resilience is unaffected by local information constraints on the routing policies.

\section{Robust equilibrium selection}
\label{sec:equilibrium-selection}
In this section, for a given flow network $\mc N$ satisfying Assumptions \ref{ass:acyclicity} and \ref{ass:flowfunction}, a constant inflow $\lambda_0\in[0,C(\mc N))$,  and locally responsive distributed routing policies with limit flow $f^*$, we shall address the issue of optimizing the strong resilience of the associated  dynamical flow network, $R(\mc N,f^*)$ with respect to $f^*$. First, in Section \ref{sec:opt1}, we shall address the issue of maximizing $R(f^*):=R(\mc N,f^*)$ over all admissible equilibrium flow vectors $f^*\in\mc F^*(\lambda_0)$, i.e., with the only constraints given by the link capacities and the conservation of mass. Then, in Section \ref{sec:opt2} we shall focus on the transportation network case, and address the problem of implementing a desired $f^*$, assuming that $f^*$ satisfies the additional constraint of being an equilibrium influenced by some static tolls. In Section \ref{sec:opt3}, we shall evaluate the gap between the maximum of $R(f^*)$ over all $f^*$, and a generic equilibrium $f^*$, and interpret it as the robustness price of anarchy with respect to $f^*$. We then distinguish between $R(f^*)$ and the commonly used metric of average delay associated to $f^*$, and then propose a convex optimization problem to solve for $f^*$ that takes into account average delay as well as strong resilience.

\subsection{Robust equilibrium flow selection as an optimization problem}\label{sec:opt1}
The robust equilibrium flow selection problem can be posed as an optimization problem as follows:
\begin{equation}
\label{eq:robust-eqm}
  R^*:=\sup_{f^*\in\mc F^*(\lambda_0)}R(f^*)\,,\end{equation}
where we recall that $\mc F^*(\lambda_0)$ is the set of admissible equilibrium flow vectors corresponding to the inflow $\lambda_0 \in [0,C(\mc N))$. 
Equation~\eqref{gammadef} implies that $ R(\floweq)$ is the minimum of a set of functions linear in $\floweq$, and hence is concave in $\floweq$. Since the closure of the constraint set $\mc F^*(\lambda_0)$ is a polytope, we get that the optimization problem stated in \eqref{eq:robust-eqm} is equivalent to a simple convex optimization problem. 
However, note that the objective function, $ R(\floweq)$ is non-smooth and one needs to use sub-gradient techniques, e.g., see \cite{Bertsekas:99}, for finding the optimal solution. 

\subsection{Using tolls for equilibrium implementation in transportation networks}\label{sec:opt2}
We now study the use of static tolls to influence the decisions of the drivers in order to get a desired emergent equilibrium condition for (unperturbed) transportation networks. The static tolls affect the driver decisions over a slower time scale at which the drivers update their preferences for global paths through the network. These global decisions are complemented by the \emph{fast-scale} node-wise route choice decisions characterized by Definition \ref{def:distributedroutingpolicy} and \ref{def:myopicpolicy}. The details of the analysis of transportation networks with such two time-scale driver decisions can be found in our companion paper~\cite{Como.Savla.ea:Wardrop-arxiv}. In particular, we show that when the time scales are sufficiently separated apart, then the network densities are attracted to a neighborhood of Wardrop equilibrium.
In this section, in order to highlight the relationship between static tolls and the resultant equilibrium point, we assume that the fast scale dynamics equilibrates quickly and focus only on the slow scale dynamics. 

We briefly describe the congestion game framework for transportation networks to formalize the equilibrium corresponding to the slow scale driver decision dynamics.
Let $\toll \in\mc R$ be the link-wise vector of tolls, with $\toll_e$ denoting the toll on link $e$. 
Assuming that $\toll$ is rescaled in such a way that one unit of toll corresponds to a unit amount of delay, the utility of a driver associated with link $e$ when the flow on it is $f_e$ is $-\left(T_e(f_e)+\toll_e \right)$, where $T_e(f_e)$ is the delay on link $e$ when the flow on it is $f_e$.
In order to formally describe the functions $T_e(f_e)$, we shall assume that each flow function $\mu_e$ satisfies Assumption~\ref{ass:flowfunction}, and additionally is strictly concave and satisfies $\mu_e'(0) < +\infty$.  Observe that the flow function described in Example~\ref{example:flowfunction} satisfies these additional assumptions. Since the flow on a link is the product of speed and density on that link, one can define the link-wise delay functions $T_e(f_e)$ by
\begin{equation}
\label{eq:delay-function}
\delayfunc_e(\flowsymb_e):= \left\{\begin{array}{ll} 
+\infty & \mbox{if } f_e \ge \flowmaxe, \\
     \mu_e^{-1}(f_e)/f_e & \mbox{if } f_e \in (0,\flowmaxe),\\
   1/\mu_e'(0) & \mbox{if } f_e = 0,
\end{array}  \right. \quad \forall e \in \edgeset.
\end{equation}
Let $\mc P$ be the set of distinct \emph{paths} from node $0$ to node $n$.
The utility associated with a path $p \in \mc P$ is $-\sum_{e \in \path} \left( \delayfunc_e\left(\flowsymb_e \right) + \toll_e \right)$.
Let $T(f)=\{T_e(f_e): \, e \in \mc E\}$ be the vector of link-wise delay functions.
We are now ready to define a \emph{toll-induced} equilibrium.
\begin{definition}[Toll-induced equilibrium]
\label{def:eqm}
For a given $\toll \in \mc R$, a toll-induced equilibrium is a vector $\floweq(\toll) \in \mc F^*$ that satisfies the following for all $p \in \paths$: 
\begin{equation*}
f_e > 0 \quad \forall e \in p \Longrightarrow \sum_{e \in \path} \left( \delayfunc_e\left(\flowsymb_e \right) + \toll_e \right) \leq \sum_{e \in q} \left( \delayfunc_e\left(\flowsymb_e \right) + \toll_e \right) \quad \forall q \in \mc P.
\end{equation*}
\end{definition}\medskip
Note that, $\floweq(\zerobf)$ corresponds to a Wardrop equilibrium, e.g., see \cite{Wardrop:52,Beckmann.McGuire.ea:56}, where $\zerobf$ is a vector all of whose entries are zero. For brevity in notation, we shall denote the Wardrop equilibrium by $\floweqwardrop$.
The following result guarantees the existence and uniqueness of a toll-induced equilibrium.

\begin{proposition}[Existence and uniqueness of toll-induced equilibrium]
\label{prop:unique}
Let $\mc N$ be a flow network satisfying Assumptions~\ref{ass:acyclicity} and \ref{ass:flowfunction} and $\lambda_0 \in [0,C(\mc N))$ a constant inflow. Assume additionally that the flow function $\mu_e$ is strictly concave and satisfies $\mu_e'(0) < +\infty$ for every link $e\in\mc E$. Then, for every toll vector $\toll \in\mc R$, there exists a unique toll-induced equilibrium $f^*(\toll) \in \mc F^*$. 
\end{proposition}
\begin{proof}
It follows from Assumption~\ref{ass:flowfunction}, strict concavity and the assumption $\mu_e'(0) < +\infty$ on the flow functions that, for all $e \in \mc E$, the delay function $\delayfunc_e(f_e)$, as defined by \eqref{eq:delay-function}, is  continuous, strictly increasing, and is such that $T_e(0)>0$.
The Proposition then follows by applying Theorems 2.4 and 2.5 from \cite{Patriksson:94}. 
\end{proof}\medskip

In this subsection, to illustrate the proof of concept, we will focus on equilibrium flows $f^*$ each of whose components is strictly positive. The results for a generic $f^* \in \mc F^*(\lambda_0)$ follow along similar lines.
Let $A \in \{0,1\}^{\mc P \times \mc E}$ be the path-link incidence matrix, i.e., for all $e \in \edgeset$ and $p \in \paths$, $A_{p,e}=1$ if $e \in p$ and zero otherwise. 
Definition~\ref{def:eqm} implies that for $f^*(\toll)\in\mc R$, with $f^*_e(\toll)>0$ for all $e\in\mc E$, to be the toll-induced equilibrium corresponding to the toll vector $\toll\in\mc R$ is equivalent to $A \left(T(f^*(\toll)) + \toll \right)=\nu \onebf$, for some $\nu >0$. We shall use this fact in the next result, where we compute tolls to get a desired equilibrium.

\smallskip

\begin{proposition}[Tolls for desired equilibrium]
\label{prop:toll}
Let $\mc N$ be a flow network satisfying Assumptions~\ref{ass:acyclicity} and \ref{ass:flowfunction} and $\lambda_0 \in [0,C(\mc N))$ a constant inflow. Assume additionally that the flow function $\mu_e$ is strictly concave and satisfies $\mu_e'(0) < +\infty$ for every link $e\in\mc E$. Assume that the Wardrop equilibrium $\floweqwardrop$ is such that $\floweqwardrop_e>0$ for all $e\in\mc E$. Let $\floweq \in\mc F^*(\lambda_0)$, with $f^*_e >0$ for all $e\in\mc E$, be the desired toll-induced equilibrium flow vector. Define $\toll(f)\in\mc R$ by
\begin{equation}\label{tolldef}
\toll(f)= \left( \max_{e \in \edgeset} \frac{T_e(f_e)}{T_e(\floweqwardrop_e)} \right) T(\floweqwardrop) - T(f)\,.
\end{equation}
Then $f^*$ is the desired toll-induced equilibrium associated to the toll vector $\toll(f^*)$.
\end{proposition}
\begin{proof}
Since $f^W$ is the Wardrop equilibrium, corresponding to the toll vector $\toll = \zerobf$, we have that
\be
\label{eq:KKT-notoll}
A T(\floweqwardrop) = \nu_1 \onebf,
\ee
for some $\nu_1>0$. For $\floweq$ to be the toll-induced equilibrium associated to the toll vector $\toll\in\mc R$, one needs to find $\nu_2 >0$ such that
\be
\label{eq:KKT-toll}
A \left( T(\floweq) + \toll\right)= \nu_2 \onebf.
\ee 
Using \eqref{eq:KKT-notoll} and simple algebra, one can verify that \eqref{eq:KKT-toll} is satisfied with $\toll(\floweq)$ as defined in (\ref{tolldef}) and $\nu_2=\nu_1 \cdot \left( \max_{e \in \edgeset} \frac{T_e(f_e^*)}{T_e(\floweqwardrop_e)} \right)$.
\end{proof}

\smallskip

\begin{remark}
The toll vector yielding a desired equilibrium operating condition is not unique. In fact, any toll of the form $\toll(\floweq)=cT(\floweqwardrop) - T(\floweq)$, with $c \geq \max\{T_e(\floweqe)/{T_e(\floweqwardrop_e)}:\,e \in \edgeset\}$ will induce $\floweq$ as the toll-induced equilibrium.
 Proposition~\ref{prop:toll} gives just one such toll vector.
\end{remark}

\subsection{The robustness price of anarchy}\label{sec:opt3}
Conventionally, transportation networks have been viewed as static flow networks, where a given equilibrium traffic flow is the outcome of driver's selfish behavior in response to the delays associated with various paths and the incentive mechanisms in place. The price of anarchy~\cite{Roughgarden:05} has been suggested as a metric to measure how sub-optimal a given equilibrium is with respect to the societal optimal equilibrium, where the societal optimality is related to the average delay faced by a driver. In the context of robustness analysis of transportation networks, it is natural to consider societal optimality from the robustness point of view, thereby motivating a notion of the robustness price of anarchy. Formally, for a $\floweq \in \mc F^*(\lambda_0)$, define the robustness price of anarchy as $\priceofanarchy \left(\floweq \right):=  R^* -  R \left(\floweq \right)$.
It is worth noting that, for a parallel topology, we have that $ R^* = R\left(\floweq \right)=\sum_{e \in \edgeset} \flowmaxe - \lambda_0$ for all $\floweq$. That is, the strong resilience is independent of the equilibrium operating condition and hence, for a parallel topology, $\priceofanarchy \left(\floweq \right) \equiv 0$. However, for a general topology and a general equilibrium, this quantity is non-zero. This can be easily justified, for example, for robustness price of anarchy with respect to the Wardrop equilibrium:
a Wardrop equilibrium is determined by the delay functions $T_e(f_e)$ as well as the topology of the network, whereas the maximizer of $R(f^*)$ depends only on the topology and the link-wise flow capacities of the network, as implied by the optimization problem in \eqref{eq:robust-eqm}. 
In fact, as the following example illustrates, for a non-parallel topology, the robustness price of anarchy with respect to Wardrop equilibrium can be arbitrarily large.
\begin{example}[Arbitrarily large robustness price of anarchy with respect to Wardrop equilibrium]
\label{example:anarchy}
Consider the network topology shown in Figure~\ref{fig:A5-justification}. Let the link-wise flow functions be given by Equation~\eqref{example:flowfunction}. The delay function is then given by $T_e(0)=\left(a_e \flowmaxe\right)^{-1}$, $T_e(f_e)=-\frac{1}{a_e f_e} \log (1-f_e/\flowmaxe)$ for $f_e \in (0,\flowmaxe)$ and $T_e(f_e)=+\infty$ for $f_e \ge \flowmaxe$.
Fix some $\epsilon \in (0,1)$ and let $\lambda_0=1/\epsilon$. Let the parameters of the flow functions be given by $\flowmax_{e_1}=\flowmax_{e_2}=1/\epsilon+\epsilon$, $\flowmax_{e_3}=\flowmax_{e_4}=1/(2 \epsilon) + \epsilon/2$, $a_1=1$, $a_2=a_3=a_4=\left(\frac{3 \epsilon}{1-\epsilon}\right) \log\left(\frac{\epsilon + \epsilon^2}{1+\epsilon^2} \right)/\log \left( \frac{1+\epsilon^2-\epsilon}{1+\epsilon^2}\right)$. For these values of the parameters, one can verify that the unique Wardrop equilibrium is given by $\floweqwardrop=[1 \quad 1/\epsilon-1 \quad 1/(2 \epsilon)-1/2 \quad 1/(2 \epsilon)-1/2]^T$. The strong resilience of $f^W$ is then given by $R(\mc N,\floweqwardrop)=\min\{2/\epsilon+2 \epsilon - 1/\epsilon, 1/\epsilon+\epsilon - (1/\epsilon-1)\}=1+\epsilon$. One can also verify that, for this case, $R^*=1/\epsilon+2 \epsilon$ which would correspond to $\floweq=[1/\epsilon \quad 0 \quad 0 \quad 0]^T$. Therefore, $P(\floweqwardrop)=1/\epsilon + 2 \epsilon - (1+ \epsilon)=1/\epsilon+\epsilon-1$ which tends to $+ \infty$ as $\epsilon \to 0^+$.
\end{example}

The above example provides a strong motivation to take robustness into account while selecting the equilibrium operating condition for the network. However, conventionally, the equilibrium selection problem for transportation networks has been primarily motivated from the point-of-view of  minimizing average delay. The average delay associated with an equilibrium $f^*$ is defined as:
\begin{equation}
\label{eq:average-delay}
D(f^*):=\sum_{e \in \mc E}f_e^* T_e(f^*_e)/\lambda_0. 
\end{equation}
The following simple example illustrates that the maximizers of $-D(f^*)$ and $R(f^*)$ are not necessarily the same.

\begin{example}
\label{example:robust-optimization}
Consider the network topology shown in Figure~\ref{fig:A5-justification}. Let the link-wise flow functions be given by Equation~\eqref{example:flowfunction}. Let the parameters of the flow function be given by: $a_{e_1}=0.01$, $a_{e_2}=a_{e_3}=a_{e_4}=10$ and $\flowmax_{e_1}=\flowmax_{e_2}=2$, $\flowmax_{e_3}=\flowmax_{e_4}=0.75$. Let $\lambda_0=2$.
The equilibrium maximizing $R(f^*)$ is $\floweq=[2 \quad 0 \quad 0 \quad 0]^T$ and the maximum strong resilience is found to be $R^*=1.5$. The minimum value of $D(f^*)$ over all $f^* \in \mc F^*(\lambda_0)$ is $15.17$, and the corresponding equilibrium $f^*$ and the value of strong resilience are $[0.5 \quad 1.5 \quad 0.75 \quad 0.75]^T$ and $0.5$ respectively. Note that the maximizers of $-D(f^*)$ and $R(f^*)$ are not necessarily the same. Therefore, a reasonable optimization problem should take into account average delay as well as network resilience. Accordingly, we propose a modified optimization problem as follows:
\begin{equation}
\label{eq:new-optimization}
  \begin{split}
    \minimize  \enspace  & D(f^*) \\
    \subj      \enspace & f^* \in \mc F^*(\lambda_0), \\
    \enspace & R(f^*) \geq b,
    \end{split}
\end{equation}
where $b \in [0,R^*]$.
Assumption~\ref{ass:flowfunction} and Equation~\eqref{eq:average-delay} imply that $D(f^*)$ is convex. Therefore, taking into account the expression for $R(f^*)$, \eqref{eq:new-optimization} is still a convex optimization problem.
Figure~\ref{fig:optim-plot} plots the outcome of this optimization as $b$ is varied from $0$ to $R^*$.
In all the cases, we solved \eqref{eq:new-optimization} using \texttt{CVX}, a package for specifying and solving convex programs \cite{CVX}. 
\end{example}

\begin{figure}
\begin{center}
\subfigure[] {\includegraphics[width=7cm,height=5.5cm]{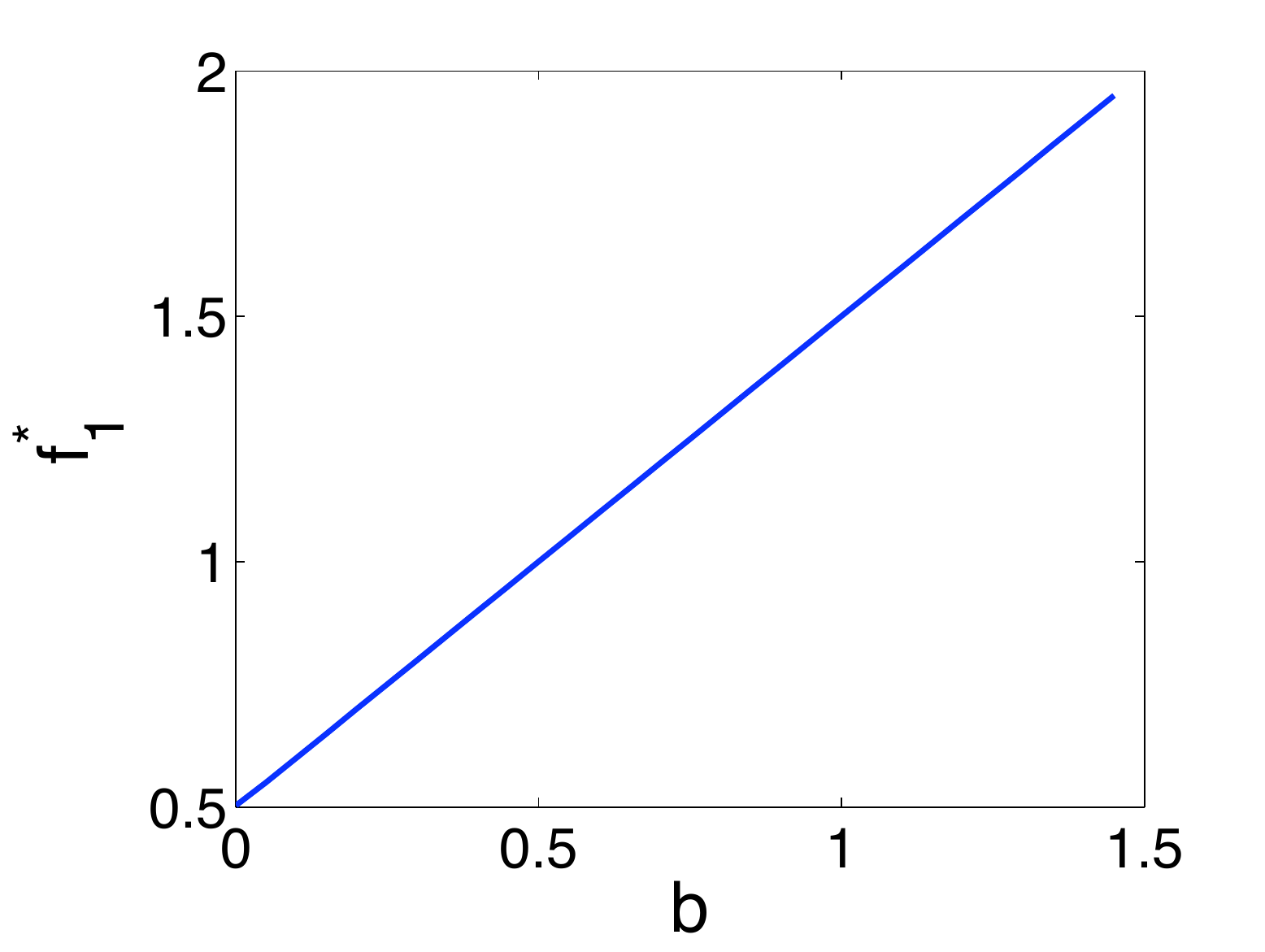}} 
\subfigure[] {\includegraphics[width=7cm,height=5.5cm]{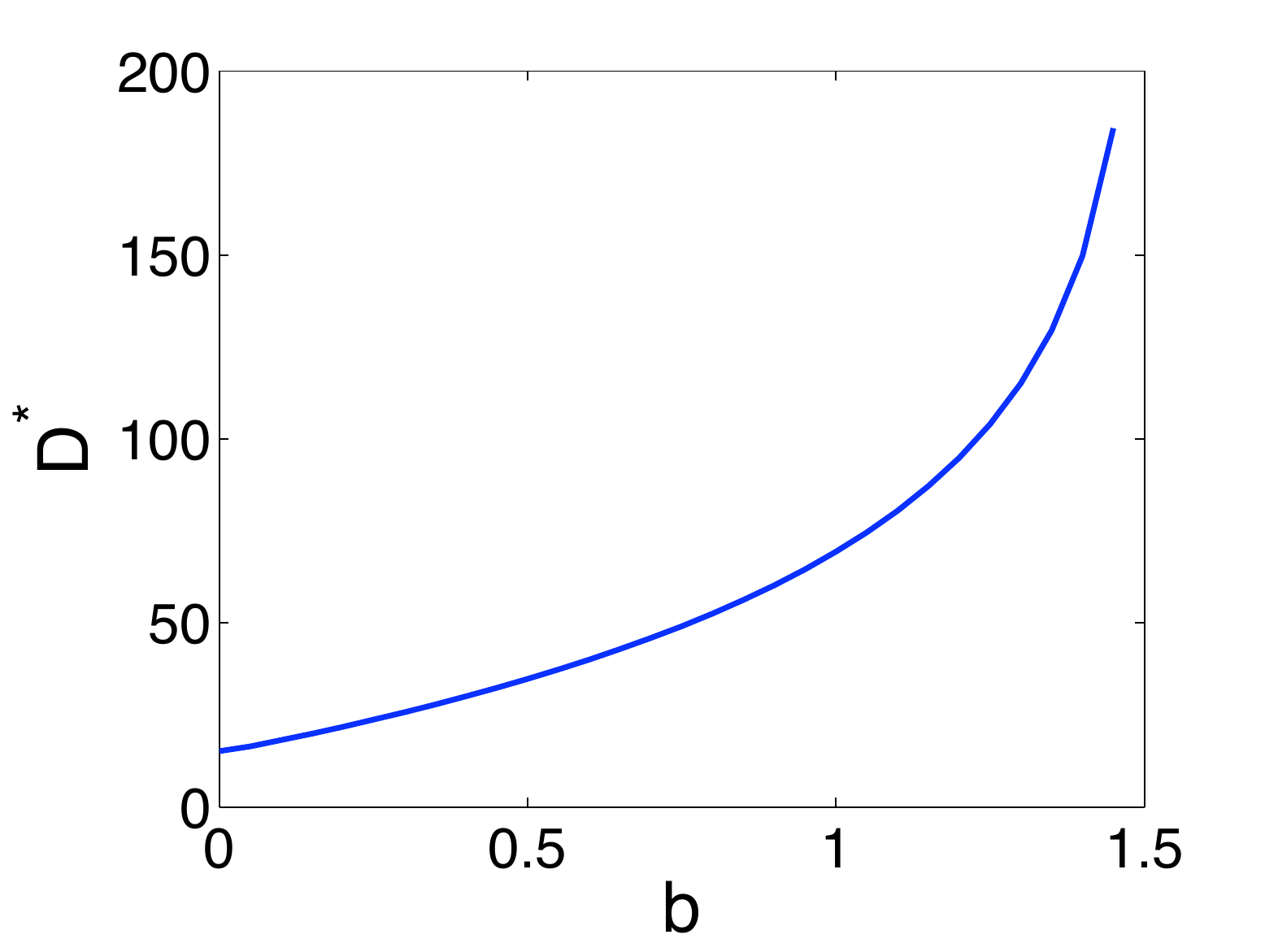}}
\end{center}
\caption{\label{fig:optim-plot} Plots of the solution of the optimization in \eqref{eq:new-optimization} for parameters specified in Example~\ref{example:robust-optimization}, as $b$ is increased from $0$ to $R^*=1.5$: (a) $f^*_1$ is the flow on link $e_1$ corresponding to $f^*$ optimizing \eqref{eq:new-optimization}; note that $f_2^*=\lambda_0-f_1^*$, and $f_3^*=f_4^*=f_2^*/2$, (b) $D^*$ is the solution of \eqref{eq:new-optimization}.}
\end{figure}


\section{Cascaded failures}
\label{sec:sim}
In this section, through numerical experiments, we study the case when the flow functions are set to the ones commonly accepted in the transportation literature, e.g., see \cite{Garavello.Piccoli:06}. In transportation literature, the flow functions are defined over a finite interval of the form $[0,\densitymaxe]$, where $\densitymaxe$ is the maximum traffic density that link $e$ can handle.
Additionally, $\mu_e$ is assumed to be strictly concave and achieves its maximum in $(0,\densitymaxe)$. 
For example, consider the following:

\begin{equation}
\label{eq:mu-func-sim-example}
\mu_e(\rho_e)= \frac{4\flowmaxe \rho_e (\densitymaxe-\rho_e) }{(\densitymaxe)^2}, \quad \rho_e \in [0,\densitymaxe].\end{equation}
An important implication of the finite capacity on the traffic densities is the possibility of cascaded \emph{spill-backs} traveling upstream as follows. When the density on a link reaches its capacity, its outflow permanently becomes zero and hence the link is effectively cut out from the network. When all the outgoing links from a particular node are cut out, it makes the outflow on all the incoming links to that node zero. Eventually, these \emph{upstream} links might possibly reach their capacity on the density and cutting themselves off permanently and cascading the effect further upstream. We shall show how such cascaded effects possibly reduce the resilience.

Another important differentiating feature of the flow functions given by \eqref{eq:mu-func-sim-example} with respect to the flow functions satisfying Assumption~\ref{ass:flowfunction} is that the flow functions corresponding to \eqref{eq:mu-func-sim-example} are not strictly increasing. As a result, one cannot readily claim that the locally responsive distributed routing policies are maximally robust for this case. However, we illustrate via simulations that, with additional assumptions,  the locally responsive distributed routing policies considered in this paper could possibly be maximally robust. 
In these simulations, we also study the effect of the flow functions given by \eqref{eq:mu-func-sim-example} on the \emph{weak resilience} of the network, which was formally defined in \cite{PartI}. In simple words, weak resilience of the network is defined as the infimum sum of the link-wise magnitude of all the disturbances under which the outflow from the destination node is asymptotically zero. In \cite[Proposition 1]{PartI}, we showed that the weak resilience of the dynamical flow network with the flow functions satisfying Assumption~\ref{ass:flowfunction} is upper bounded by its  min-cut capacity. It is easy to show that this upper bound on weak resilience also holds when the flow functions are the ones given by \eqref{eq:mu-func-sim-example}.
 
For the simulations, we selected the following parameters:
\begin{itemize}
\item the graph topology $\graph$ shown in Figure~\ref{fig:graph}.

\begin{figure}
\begin{center}
\scalebox{1.0}
{
\begin{tikzpicture}

\path (0,0) node(a) [circle,draw,fill=blue!50!white] {$0$}
		(2,1.5) node (b) [circle,draw,fill=blue!50!white] {$1$}
	(2, 0) node (c) [circle,draw,fill=blue!50!white] {$2$}
   (2, -1.5) node (d) [circle,draw,fill=blue!50!white] {$3$}
   (4, 1) node (e) [circle,draw,fill=blue!50!white] {$4$}
   (4, -1) node (f) [circle,draw,fill=blue!50!white] {$5$}
    (6, 2.5) node (g) [circle,draw,fill=blue!50!white] {$6$}
    (6, -2.5) node (h) [circle,draw,fill=blue!50!white] {$7$}
    (8, 0) node (i) [circle,draw,fill=blue!50!white] {$8$};
	
\draw[thick,->] (-1,0) -- (a); 
\draw[very thick,->] (a) -- (b);
\draw[very thick,->] (a) -- (c);
\draw[very thick,->] (a) -- (d);
\draw[very thick,->] (b) -- (e);
\draw[very thick,->] (c) -- (e);
\draw[very thick,->] (c) -- (f);
\draw[very thick,->] (d) -- (f);
\draw[very thick,->] (e) -- (i);
\draw[very thick,->] (f) -- (i);
\draw[very thick,->] (e) -- (g);
\draw[very thick,->] (f) -- (h);
\draw[very thick,->] (b) -- (g);
\draw[very thick,->] (d) -- (h);
\draw[very thick,->] (g) -- (i);
\draw[very thick,->] (h) -- (i);

\draw (-0.65,0.3) node {$\lambda_0$};	
\draw (1,1.1) node {$e_1$};
\draw (1,0.2) node {$e_2$};	
\draw (1,-1.1) node {$e_3$};	
\draw (3.2,1.4) node {$e_4$};
\draw (2.8,0.7) node {$e_5$};	
\draw (3,-0.3) node {$e_6$};	
\draw (2.9,-1.1) node {$e_7$};	
\draw (3.8,-2.2) node {$e_8$};	
\draw (3.8,2.2) node {$e_9$};
\draw (5.2,1.6) node {$e_{10}$};	
\draw (5.2,-1.6) node {$e_{11}$};	
\draw (6,0.7) node {$e_{12}$};	
\draw (6,-0.7) node {$e_{13}$};	
\draw (7.2,1.5) node {$e_{14}$};	
\draw (7.2,-1.5) node {$e_{15}$};		
	
\end{tikzpicture}
}
\end{center}
\caption{The graph topology used in simulations.}
\label{fig:graph}
\end{figure}
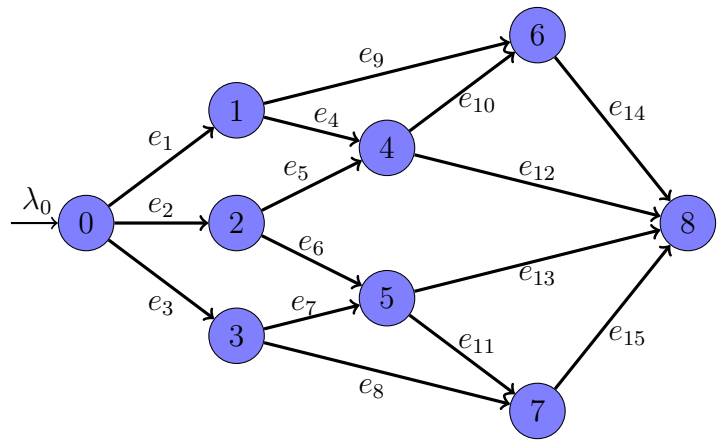

\item $\lambda_0=3$.
\item let $\densitymax_{e}=3$ for all $e\in\mc E$, and flow capacities given by 
$\flowmax_{e_1}=\flowmax_{e_2}=\flowmax_{e_3}=2.5$, $\flowmax_{e_4}=0.9$, $\flowmax_{e_5}=1.75$, $\flowmax_{e_6}=\flowmax_{e_{11}}=\flowmax_{e_{13}}=1$, $\flowmax_{e_7}=\flowmax_{e_8}=0.7$, $\flowmax_{e_9}=0.4$, $\flowmax_{e_{10}}=\flowmax_{e_{12}}=1.5$, $\flowmax_{e_{14}}=2$, and $\flowmax_{e_{15}}=1.6$. The link-wise flow functions are as given in \eqref{eq:mu-func-sim-example}, if $e \in \mc E_n^-$ or if $\rho < \rho^{\max}_{e'}$ for at least one \emph{downstream} edge $e'$, i.e., $e' \in \mc E$ such that $e \in \mc E_v^-$ and $e' \in \mc E_v^+$ for some $v \in \until{n-1}$, and the flow functions are uniformly zero otherwise;
\item the equilibrium flow $f^*$ has components $f^*_{e_1}=f^*_{e_3}=f^*_{e_6}=0.5$, $f^*_{e_2}=2$, $f^*_{e_4}=f^*_{e_{13}}=0.3$, $f^*_{e_5}=1.5$, $f^*_{e_7}=f^*_{e_8}=0.25$, $f^*_{e_9}=0.2$, $f^*_{e_{10}}=f^*_{e_{12}}=0.9$, $f^*_{e_{11}}=0.2$, $f^*_{e_{13}}=0.3$, $f^*_{e_{14}}=1.1$, and $f^*_{e_{15}}=0.7$;
\item the route choice function is as follows:
\begin{equation*}
G^v_e(\rho^v)=\frac{\floweqe \exp(-\ilogitconst(\rho_e-\rho_e^{*})) \1_{[0,\rho_e^{\mathrm{max}}]}( \rho_e) }{\sum_{j\in\mc E_v^+} \floweq_j \exp(-\ilogitconst(\rho_j-\rho_j^{*})) \1_{[0, \rho_j^{\mathrm{max}}]}(\rho_j)},
\end{equation*}
where $\ilogitconst$ will be a variable parameter for the simulations. Note that this is a modified version of the route choice function given by \eqref{eq:routing-example}. The modification is done to respect the finite traffic density constraint on the links. 
\end{itemize}

One can verify that, with these parameters, the minimum node residual capacity, and hence an upper bound on the strong resilience, as defined by \eqref{gammadef} is $0.75$. One can also verify that the maximum flow capacity of the network, and hence an upper bound on the weak resilience, is $5.2$.

\subsection{Effect of $\ilogitconst$ on the strong resilience}
Consider an admissible perturbation such that $\tilde \mu_{e_{10}}=\frac{8}{15} \mu_{e_{10}}$ and $\tilde \mu_{e_k} = \mu_k$ for all $k \in \{1,\ldots,15\}\setminus \{10\}$. As a result, $\delta_{e_{10}}=0.7$ and $\delta_{e_k}=0$ for all $k \in \{1,\ldots,15\}\setminus \{10\}$. Therefore, the magnitude of the perturbation is $\delta=0.7$. Note that this value is less than the minimum node residual capacity of the network.
We found that $\lim_{t \to \infty} \lambda_{e_8}(t)=0$ for all $\ilogitconst < 0.25$, and $\lim_{t \to \infty} \lambda_{e_8}(t)=\lambda_0=3$ for all $\ilogitconst \geq 0.25$. 
The role of $\ilogitconst$ in the strong resilience is best understood by concentrating on a parallel topology consisting of edges $e_{10}$ and $e_{12}$ with arrival rate $\lambda_{e_4}^*$. Using similar techniques as in the proof of Theorem~\ref{maintheo}, one can show the existence of a new equilibrium for this \emph{local} system. However, this equilibrium is not attractive from a configuration where at least one of $\tilde \rho_{e_{10}}$ or $\tilde \rho_{e_{12}}$ is at $\rho^{\max}_{e_{10}}$ or $\rho^{\max}_{e_{12}}$, respectively. For $\ilogitconst < 0.25$, $\tilde \rho_{e_{10}}$ reaches $\rho^{\max}_{e_{10}}$, whereas for $\ilogitconst \ge 0.25$, neither $\tilde \rho_{e_{10}}$ nor $\tilde \rho_{e_{12}}$ hit the maximum density capacity and the system is attracted towards the new equilibrium.

\subsection{Effect of cascaded shutdowns on the weak resilience}
Consider an admissible disturbance such that $\tilde \mu_{e_4}=\frac{2}{9} \mu_{e_4}$, $\tilde \mu_{e_5} = \frac{23}{35} \mu_{e_5}$, $\tilde \mu_{e_6}=\frac{4}{5} \mu_6$, $\tilde \mu_{e_7}=\frac{2}{7} \mu_{e_7}$, $\tilde \mu_{e_8}=\frac{2}{7} \mu_{e_8}$, $\tilde \mu_{e_9}=\frac{1}{2} \mu_{e_9}$, $\tilde \mu_{e_{10}}=\frac{3}{5} \mu_{e_{10}}$, $\tilde \mu_{e_{12}}=\frac{8}{15} \mu_{e_{12}}$ and $\tilde \mu_k = \mu_k$ for $k=\{1,2,3,11,13,14,15\}$.
As result, $\delta_{e_4}=0.7$, $\delta_{e_5}=0.6$, $\delta_{e_6}=0.2$, $\delta_{e_7}=0.5$, $\delta_{e_8}=0.5$, $\delta_{e_9}=0.2$, $\delta_{e_{10}}=0.6$, $\delta_{e_{12}}=0.7$ and $\delta_{e_k}=0$ for $k=\{1,2,3,11,13,14,15\}$. Therefore, $\delta=4$, which is less than the min-cut flow capacity of the network. For this case, it is observed that, $\lim_{t \to \infty} \lambda_{e_8}(t)=0$ independent of the value of $\ilogitconst$. 
This can be explained as follows. For the given disturbance, we have that $\tilde f^{\max}_{e_{10}}+\tilde f^{\max}_{e_{12}} = 1.7 < 1.8 = f^*_{e_{10}} + f^*_{e_{12}}$. Therefore, after finite time $t_1$, we have that $\tilde \rho_{e_{10}}(t)=\rho^{\max}_{e_{10}}$ and $\tilde \rho_{e_{12}}(t)=\rho^{\max}_{e_{12}}$ for all $t \geq t_1$. As a consequence, we have that, $\tilde f_{e_{4}}(t)=0$ and $\tilde f_{e_{5}}(t)=0$ for all $t \geq t_1$. One can repeat this argument to conclude that, for the given disturbance, after finite time, $\tilde \rho_{e_{k}}$ for $k=1,\ldots,9$ reach and remain at their maximum density capacities. As a consequence, after such a finite time, $\tilde f_{e_1}(t)+\tilde f_{e_2}(t)+\tilde f_{e_3}(t)=0$ and hence, $\lim_{t \to \infty} \lambda_{e_8}(t)=0$, i.e., the network is not partially transferring. This is also illustrated in Figure~\ref{fig:cascade} which plots the flow through some of the links of the network as a function of time.
This example illustrates that the cascaded effects can potentially reduce the weak resilience of a dynamical flow network.

\begin{figure}
\begin{center}
\subfigure {\includegraphics[width=5.4cm,height=4.8cm]{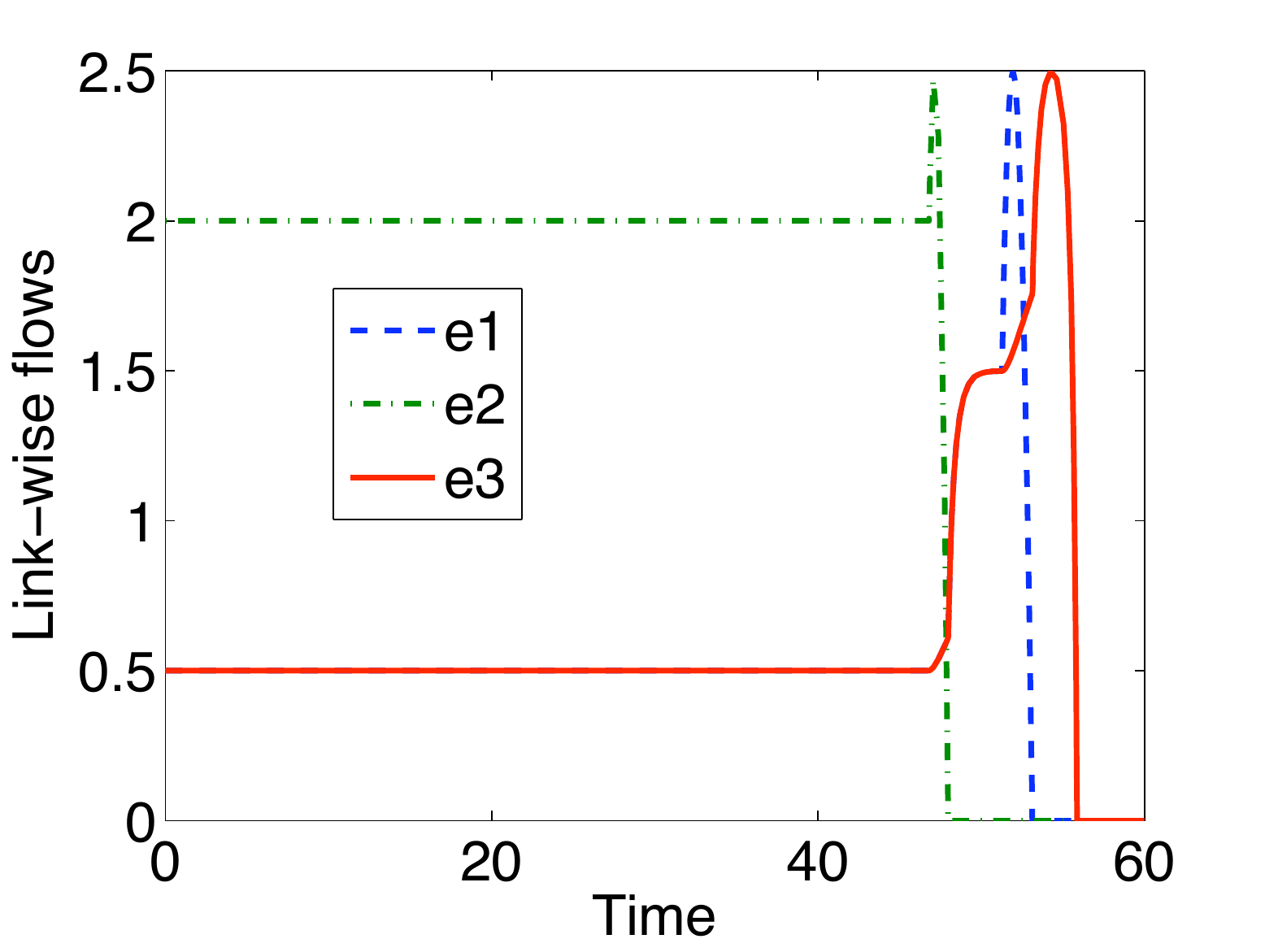}} 
\subfigure {\includegraphics[width=5.4cm,height=4.8cm]{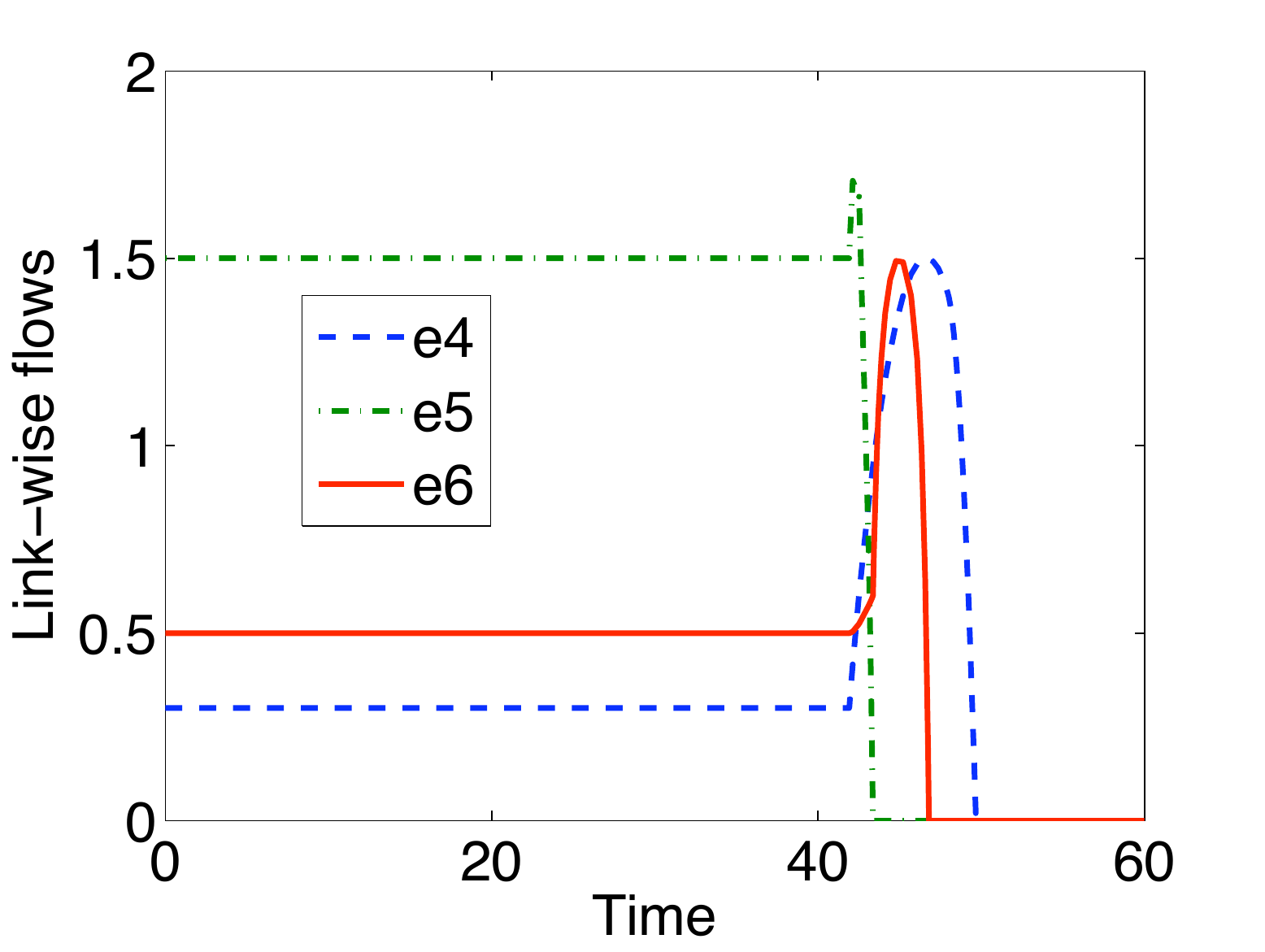}}
\subfigure {\includegraphics[width=5.4cm,height=4.8cm]{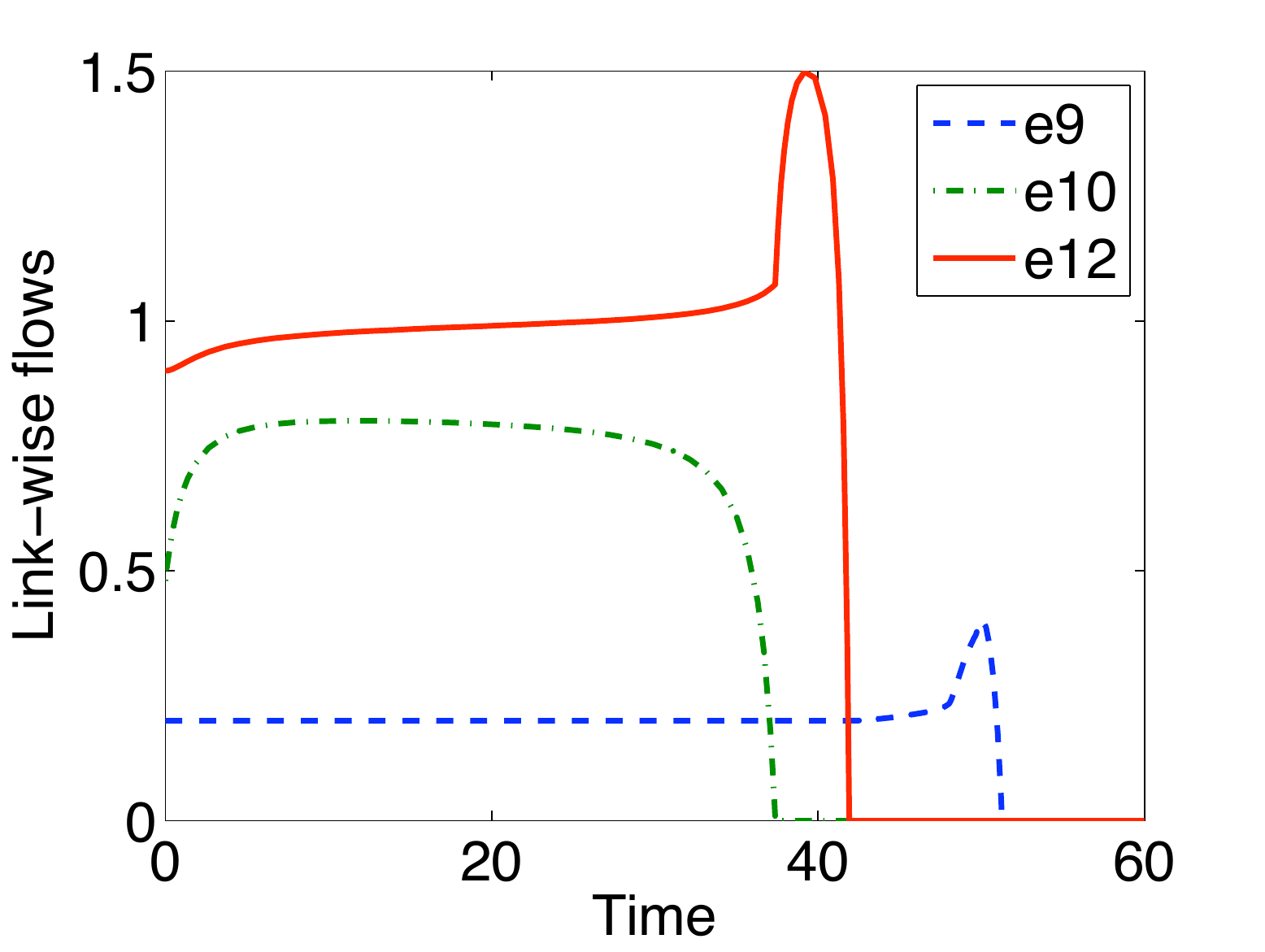}}
\end{center}
\caption{\label{fig:cascade} Plot of link-wise flows for some of the links of the network that ultimately shut down. The timings of shut downs of the links demonstrate the cascaded effect starting from link $e_{10}$ and traveling up to the origin node.}
\end{figure}

\section{Conclusion}
\label{sec:conclusion}
In this paper, we studied strong resilience of dynamical flow networks, with respect to perturbations that reduce the flow functions of the links of the network. We showed that locally responsive distributed routing policies yield the maximum strong resilience under local information constraint. We also showed that the corresponding strong resilience is equal to the minimum node residual capacity of the network, and hence depends on the limit flow of the unperturbed network.
Our results show that, unlike the weak resilience which was considered in \cite{PartI}, the strong resilience of a dynamical flow network is sensitive to local information constraint. 
We proposed simple convex optimization problems to solve for equilibria that maximize traditional metrics of social optimality such as average delay subject to guarantees on strong resilience. We also discussed the use of tolls to induce a generic equilibrium flow for the unperturbed system in the context of transportation networks. Finally, we also discussed cascaded failures due to spill backs when we impose finite density constraints on the links and illustrated the utility of routing policies discussed in this paper in averting such failures.
The findings of this and the companion paper~\cite{PartI} stand to provide important guidelines for management of several large scale critical infrastructures both from planning as well as real-time operation point of view. 

In future, we plan to extend the research in several directions. We plan to rigorously study the robustness properties of the network with finite link-wise capacity for the densities, and formally establish the results on the resilience as suggested in Section~\ref{sec:sim}.
We plan to study the scaling of the resilience with respect to the amount of information, e.g., multi-hop as opposed to just single-hop, available to the routing policies. We also plan to perform robustness analysis in a probabilistic framework to complement the adversarial framework of this paper, possibly considering other general models for disturbances. In particular, it would be interesting to study robustness with respect to sequential disturbances than just one-shot disturbance considered in this paper. We plan to consider a setting with buffer capacities on the nodes and study the scaling of the resilience with such buffer capacities.
We also plan to consider more general graph topologies, e.g., graphs having cycles and multiple origin-destination pairs. 
 
\appendices
\section{Proof of Theorem \ref{lemma:upperbound}}\label{sec:proof1}
In this section, we shall prove Theorem \ref{lemma:upperbound} by showing that, given a flow network $\mc N$ satisfying Assumptions \ref{ass:acyclicity} and \ref{ass:flowfunction}, a constant inflow $\lambda_0\ge0$, a distributed routing policy $\mc G$, and a limit flow $f^*\in\cl(\mc F)$ for the associated dynamical flow network (\ref{dynsyst}), the strong resilience satisfies $$\gamma_1(f^*,\mc G)= R(\mc N,f^*)\,.$$

Let $f^{\circ}\in\mc B(f^*)$ be some initial flow attracted by $f^*$. In order to prove the result it is sufficient to exhibit a family of admissible perturbations, with magnitude $\delta$ arbitrarily close to $R(\mc N,f^*)$, under which the network is not fully transferring with respect to $f^{\circ}$. Let us fix some non-destination node $0\le v<n$ minimizing the right-hand side of (\ref{gammadef}), and put $\kappa:=\sum_{e\in\mc E^+_v}\flowmax_e$. For any $R(\mc N,f^*)<\delta<\kappa$, consider the admissible perturbation defined by 
 \be\label{pertflowdelta}\tilde\mu_e(\rho_e):=\frac{\kappa-\delta}{\kappa}\mu_e(\rho_e)\,,\quad\forall e\in\mc E^+_v\,,\qquad\tilde\mu_e(\rho_e):=\mu_e(\rho_e)\,,\quad\forall e\in\mc E\setminus\mc E^+_v\,. \ee 
Clearly, the magnitude of such perturbation equals $\delta$.

Let us consider the origin-destination cut-set $\mc U:=\{0,1,\ldots,v\}$, and put $$\mc E^+_{\mc U}:=\{(u,w)\in\mc E:\,0\le u\le v, v<w\le n\}\,.$$ Observe that, thanks to Assumption \ref{ass:acyclicity} on the acyclicity of the network topology, since all the edges outgoing from some node $u\le v$ are unaffected by the perturbation,
the associated perturbed dynamical flow network (\ref{pertdynsyst}) with initial flow $\tilde f(0)=f^{\circ}\in\mc B(f^*)$ satisfies $$\lim_{t\to+\infty}\tilde f_e(t)=\lim_{t\to+\infty}f_e(t)= \floweqe\,,\qquad\forall e\in\mc E^+_u\,,\quad\forall 0\le u<v\,.$$
In particular, this implies that $\tilde{\mu}_e(\tilde \rho_e(t))= \floweqe$ for all $t\ge0$, and for every link $e\in\mc E^+_{\mc U}\setminus\mc E^+_v$. On the other hand, one has that $$\tilde f_e(t)<\tilde f_e^{\max}=\frac{\kappa-\delta}{\kappa}\flowmaxe\,,\qquad\forall e\in\mc E^+_v\,,\ \forall t\ge0\,.$$ Therefore, one has that  
\be\label{eq:cutflowbound}
\ba{rcl}\ds\limsup_{t\to+\infty}\sum\nolimits_{e\in\mc E^+_{\mc U}}\tilde f_e(t)&\le&
\ds\sum\nolimits_{e \in \mc E_v^+} \newflowmaxe + \sum\nolimits_{e \in \mc E_{\mc U}^+ \setminus \mc E_v^+} \floweqe \\[7pt]&= &\ds\frac{\kappa-\delta}{\kappa} \sum\nolimits_{e \in \mc E_v^+} \flowmaxe + \sum\nolimits_{e \in \mc E_{\mc U}^+ \setminus \mc E_v^+} \floweqe\\[7pt]&= &\ds\sum\nolimits_{e \in \mc E_v^+} \flowmaxe - \delta - \sum\nolimits_{e \in \mc E_v^+} f_e^* + \sum\nolimits_{e \in \mc E_{\mc U}^+} f_e^* \\[7pt]
&= & R(\mc N,f^*) - \delta + \lambda_0\,.\ea\ee
Observe that, for every $v<w<n$, and $t\ge0$, 
\be\label{desumrhoedet}\ba{rcl}\ds\frac{\de}{\de t}\l(\sum\nolimits_{e\in\mc E^+_w}\tilde \rho_e(t)\r)&=&
\ds\sum\nolimits_{e\in\mc E^+_w}\l(\sum\nolimits_{e\in\mc E^-_w}\tilde f_e(t)\r)G^v_e(\tilde\rho^w(t))-\sum\nolimits_{e\in\mc E^+_w}\tilde f_e(t)\\[7pt]&=&\ds\sum\nolimits_{e\in\mc E^-_w}\tilde f_e(t)-\sum\nolimits_{e\in\mc E^+_w}\tilde f_e(t)\,.\ea\ee
Define the edge sets $$\mc A:=\bigcup\nolimits_{w=v+1}^{n-1}\mc E^+_w\,,\qquad 
\mc D:=\bigcup\nolimits_{w=v+1}^{n}\mc E^-_w\,,$$ 
and put $\zeta(t):=\sum_{e\in\mc A}\rho_e(t)$. 
Using (\ref{desumrhoedet}), the identity $\mc A\cup\mc E^+_{\mc U}=\mc D$, and (\ref{eq:cutflowbound}), one gets that there exists some $\tau'\ge0$ such that 
\be\label{dezetadet}\ba{rcl}\ds\frac{\de}{\de t}\zeta(t)&=&
\ds\sum\nolimits_{v<w\le n}\sum\nolimits_{e\in\mc E^-_w}\tilde f_e(t)-\ds\sum\nolimits_{v<w\le n}\sum\nolimits_{e\in\mc E^+_w}\tilde f_e(t)\\[7pt]&=&\ds\sum\nolimits_{e\in\mc D}\tilde f_e(t)-\sum\nolimits_{e\in\mc E^-_n}\tilde f_e(t)-\sum\nolimits_{e\in\mc A}\tilde f_e(t)\\[7pt]&=&\ds\sum\nolimits_{e\in\mc E^+_{\mc U}}\tilde f_e(t)-\sum\nolimits_{e\in\mc E^-_n}\tilde f_e(t)\\[7pt]&\le&\ds R(\mc N,f^*) - \delta + \lambda_0-\tilde\lambda_n(t)+\eps\,,
\ea\ee
for all $t\ge\tau'$. Now assume, by contradiction, that 
$$\liminf_{t\to+\infty}\tilde\lambda_n(t)>  R(\mc N,f^*) - \delta + \lambda_0\,.$$
Then, there would exist some $\eps>0$ and $\tau''\ge0$ such that $$\tilde\lambda_n(t)\ge R(\mc N,f^*) - \delta + \lambda_0+2\eps\,,\qquad t\ge\tau''\,.$$ It would then follow from (\ref{dezetadet}) and Gronwall's inequality that 
$$\zeta(t)\le\zeta(\tau)-(t-\tau)\eps\,,\qquad \forall t\ge\tau\,,$$ where $\tau:=\max\{\tau',\tau''\}$. Then, $\zeta(t)$ would converge to $-\infty$ as $t$ grows large, contradicting the fact that $\zeta(t)\ge0$ for all $t\ge0$. Hence, necessarily
$$\liminf_{t\to+\infty}\tilde\lambda_n(t)\le  R(\mc N,f^*) - \delta + \lambda_0<\lambda_0\,,$$
so that the perturbed dynamical flow network is not fully transferring.
Then, from the arbitrariness of the perturbation's magnitude $\delta\in(R(\mc N,f^*),\kappa)$, it follows that the network's strong resilience is upper bounded by $R(\mc N,f^*)$.

\section{Proof of Theorem \ref{maintheo}}\label{sec:proof2}
In this section, we shall prove Theorem \ref{maintheo}, by showing that, given a flow network $\mc N$ satisfying Assumptions \ref{ass:acyclicity} and \ref{ass:flowfunction}, a constant inflow $\lambda_0\ge0$, and a locally responsive distributed routing policy $\mc G$, then the strong resilience of the unique limit flow $f^*\in\cl(\mc F)$ of the associated dynamical flow network (\ref{dynsyst}) satisfies $$\gamma_1(f^*,\mc G)= R(\mc N,f^*)\,.$$ 
Thanks to Theorem \ref{lemma:upperbound}, it is sufficient to show that 
\be\label{gamma1LB}\gamma_1(f^*,\mc G)\ge R(\mc N,f^*)\,.\ee 

First, let us consider the case when $f^*\in\cl(\mc F)\setminus\mc F^*(\lambda_0)$, i.e., when the limit flow of the unperturbed dynamical flow network (\ref{dynsyst}) is not an equilibrium. As argued in Remark \ref{remark:limitflowequilibrium}, in this case some of the capacity constraints are satisfied with equality, i.e., there exist $0\le v<n$ and $e\in\mc E^+_v$ such that $f^*_{e}=f^{\max}_e$. Then, Theorem \ref{thm:uniquelimitflow} implies that $f^*_{e}=f^{\max}_e$ for all $e\in\mc E^+_v$, so that 
$$R(\mc N,f^*)\le\sum_{e\in\mc E^+_v}\l(f^{\max}_e-f_e^*\r)=0\,,$$
and (\ref{gamma1LB}) is trivially satisfied, since $\gamma_1(f^*,\mc G)\ge0$ by definition. Therefore, for the rest of this section, we shall restrict ourselves on the case when $f^*\in\mc F^*(\lambda_0)$, i.e., when $f^*$ is a globally attractive equilibrium flow of the unperturbed dynamical flow network (\ref{dynsyst}). 

Observe that, for any admissible perturbation, regardless of its magnitude, the perturbed dynamical flow network (\ref{pertdynsyst}) satisfies all the assumptions of Theorem \ref{thm:uniquelimitflow}, which can therefore be applied to show the existence of a globally attractive perturbed limit flow $\tilde f^*\in\mc\cl(\mc F)$. This in particular implies that $\tilde\lambda_n(t)=\sum_{e\in\mc E^-_n}\tilde f_e(t)$ converges to $\tilde\lambda_n^*=\sum_{e\in\mc E^-_n} \tilde f^*_e$ as $t$ grows large. However, this is not sufficient in order to prove strong resilience of the perturbed dynamical flow network (\ref{pertdynsyst}), as it might be the case that $\tilde\lambda^*_n<\lambda_0$.

In fact, it turns out that, if the magnitude of the admissible perturbation is smaller than $R(\mc N,f^*)$, the perturbed limit flow $\tilde f^*$ is an equilibrium flow for the perturbed dynamical flow network, so that $\tilde\lambda^*_n=\lambda_0$ and (\ref{pertdynsyst}) is fully transferring. In order to show this, we need to study the \emph{perturbed local  system} 
\be\label{pertlocalsys}
\frac{\de}{\de t}\tilde\rho_e(t)=\tilde\lambda(t)G^v_e(\tilde\rho^v(t))-\tilde f_e(t)\,,\qquad \tilde f_e(t)=\tilde\mu_e(\tilde\rho_e(t))\,,\qquad\forall e\in\mc E^+_v
\,,\ee
for every non-destination node $0\le v<n$, and nonnegative-real-valued, Lipschitz continuous local input $\tilde\lambda(t)$. Indeed, \cite[Lemma 4]{PartI} can be applied to the perturbed local system (\ref{pertlocalsys}) establishing convergence of the perturbed local flows $\tilde f^v(t)$ to a local equilibrium flow $\tilde f^*(\lambda)\in\mc F_v$, provided that the input flow $\tilde\lambda(t)$ converges, as $t$ grows large, to a value $\lambda$ which is strictly smaller than the sum of the perturbed flow capacities of the outgoing links. However, such local result is not sufficient to prove strong resilience of the entire perturbed dynamical flow network. The key property in order to prove such a global result is stated in Lemma \ref{lemmadiffusivity}, which describes how the flow redistributes itself upon the network perturbation. In particular, such result ensures that the increase in flow on all the links downstream from a node whose outgoing links are affected by a given perturbation, is less than the magnitude of the disturbance itself. We shall refer to this property as the \emph{diffusivity} of the local perturbed system. 
\begin{lemma}[Diffusivity of the local perturbed system]
\label{lemmadiffusivity}
Let $\mc N$ be a flow network satisfying Assumptions \ref{ass:acyclicity} and \ref{ass:flowfunction}, $\mc G$ be a  locally responsive distributed routing policy, $\lambda_0\ge0$ a constant inflow. Assume that $f^*\in\mc F^*(\lambda_0)$ is an equilibrium flow for the dynamical flow network (\ref{dynsyst}). Let $\tilde{\mc N}$ be an admissible perturbation of $\mc N$, $0\le v<n$ be a nondestination node, $\lambda^*_v:=\sum_{e\in\mc E^+_v}f^*_e$, and $\lambda\in[0,\sum_{e\in\mc E^+_v}\tilde f^{\max}_e)$. Then, for every $\mc J\subseteq\mc E_v^+$,  the local equilibrium flow $\tilde f^{*}(\lambda)$ of the perturbed local system (\ref{pertdynsyst}) with constant local input $\tilde\lambda(t)\equiv\lambda$ satisfies
\be\sum\nolimits_{e \in\mc  J} \left(\tilde f^*_e(\lambda)-\floweq_e \right) \le 
\left[\lambda-\lambda_v^* \right]_+ + \sum\nolimits_{e \in \edgeset_v^+} \delta_e\,, \label{diffusivity}\ee
where $\delta_e:=||\mu_e(\,\cdot\,)-\tilde\mu_e(\,\cdot\,)||_{\infty}$.
\end{lemma}
\begin{proof}
Define $\lambda_v^{*}:=\sum\nolimits_{e \in \mc E_v^+}\floweqe$, and $\hat\lambda:=\max\{\lambda,\lambda_v^{*}\}$. Let $\hat\rho^v(t)$ be the solution of the perturbed local system (\ref{pertlocalsys}) with constant input $\tilde\lambda(t)\equiv\hat\lambda$, and initial condition $\hat\rho_e(0)=\rho^*_e:=\mu_e^{-1}(f^*_e)$, for all $e\in\mc E^+_v$, and let $\hat f_e(e):=\tilde\mu_e(\hat\rho_e(t))$. We shall first prove that 
\be\label{morethanstart}
\hat f_e(t)\ge \tilde \mu_e(\rho_e^*)\,, \quad \forall \, t\ge0\, \quad \forall \, e \in \mc E_v^+. 
\ee
For this, consider a point $\hat \rho^v\in\mc R_v$, such that $\hat \rho^v \ne\rhoeq$, and there exists some $i\in\mc E_v^+$ such that $\hat \rho_i=\rhoeq_i$ and $\hat \rho_e\ge\rhoeq_e$ for all $e\ne i\in\mc E_v^+$. For such a $\hat \rho^v$ and $i$, \cite[Lemma 4]{PartI} implies that $G^v_i(\hat \rho^v) \geq G^v_i(\rhoeq)$. This, combined with the fact that $\hat \lambda\geq \lambda_v^{*}$ and $$\tilde \mu_i(\hat \rho_i) \leq \mu_i (\hat \rho_i) = \mu_i (\rhoeq_i)\,,$$ yields
\begin{equation}
\label{eq:hat-der}
\hat\lambda_v G_i^v(\hat \rho^v)-\tilde\mu_i(\hat \rho_i) \geq \lambda_v^{*} G_i^v(\rhoeq)-\mu_i(\rhoeq_i)=0\,.
\end{equation}
Considering the region $\Omega:=\{\hat \rho^v \in\mc R_v:\,\hat \rho_j\ge\rhoeq_j\,,\ \forall j\in\mc E_v^+\}$, and denoting by $\omega\in\R^{\mc E^+_v}$ the unit outward-pointing normal vector to the boundary of $\Omega$ at $\hat \rho^v$, \eqref{eq:hat-der} shows that 
$$\frac{\de}{\de t}\hat\rho^v \cdot \omega= \l(\hat\lambda_v G^v(\hat \rho^v)-\tilde\mu_v(\hat \rho^v)\r)\cdot\omega \leq 0\,,\quad \forall \hat \rho^v \in\partial\Omega\,,\ t\ge0\,.$$
Therefore, $\Omega$ is invariant under (\ref{pertlocalsys}). Since $\hat \rho^v(0) = \rhoeq \in \Omega$, this proves \eqref{morethanstart}.

Now, \cite[Lemma 4]{PartI} implies that there exists a unique local equilibrium flow $\hat f^*:=\tilde f^*(\hat\lambda)$. Then, for any $\mc J \subseteq \mc E_v^+$, (\ref{morethanstart}) implies that 
\be\label{eq:diff}\ba{rcl}
\ds\sum\nolimits_{j}\hat f_j^*&=&
\ds\hat\lambda_v^*-\sum\nolimits_{k}\hat f_k^* \\[7pt]
&\le&\ds\hat\lambda_v^*-\sum\nolimits_{k}\tilde\mu_k(\rhoeq_k) \\[7pt]
&=&\ds\hat \lambda_v^*-\lambda^{*}_v +\sum\nolimits_{j}\floweq_j+\sum\nolimits_{k}\mu_k(\rhoeq_k)-\sum\nolimits_{k }\tilde\mu_k(\rhoeq_k)\\[7pt]
&\le&\ds[\hat \lambda_v^*-\lambda^{*}_v]_++\sum\nolimits_{j}\floweq_j + \sum\nolimits_{k } \delta_k\\[7pt]
&\le&\ds[\hat \lambda_v^*-\lambda^{*}_v]_++\sum\nolimits_{j}\floweq_j + \sum\nolimits_{e} \delta_e\, ,
\ea\ee
where the summation indices $j$, $k$, and $e$ run over $\mc J$, $\mc E^+_v \setminus \mc J$, and $\mc E^+_v$, respectively. 
Moreover, since $\lambda\le\hat\lambda$ from \cite[Lemma 3]{PartI}, one gets that 
$\tilde f_e^*(\lambda)\leq\tilde f_e^*(\hat\lambda)=\hat f_e^*$ for all $e\in\mc E_v^+$. In particular, this implies that $$\sum_{j \in \mc J}\tilde f_j^*(\lambda)\leq\sum_{j \in \mc J}\hat f_j^*\,,\qquad \forall \mc J \subset \mc E_v^+\,.$$ This, combined with \eqref{eq:diff}, proves \eqref{diffusivity}.
\end{proof}\medskip

The following lemma exploits the diffusivity property from Lemma~\ref{lemmadiffusivity} along with an induction argument on the topological ordering of the node set to prove that $R(\mc N,f^*)$ is indeed a lower bound on the strong resilience of the network under the locally responsive distributed routing policies.

\begin{lemma}[Globally attractive equilibrium for perturbed flow network]
\label{lemmainduction}
Consider a flow network $\mc N$ satisfying Assumptions \ref{ass:acyclicity} and \ref{ass:flowfunction}, a locally responsive distributed routing policy $\mc G$, and a constant inflow $\lambda_0\ge0$. Assume that $f^*\in\mc F^*(\lambda_0)$ is an equilibrium flow for the associated dynamical flow network. Let $\tilde{\mc N}$ be an admissible perturbation of $\mc N$, of magnitude $\delta< R(\mc N,f^*)$. Then, the perturbed dynamical flow network (\ref{pertdynsyst}) has a globally attractive equilibrium flow and hence it is fully transferring.
\end{lemma}
\begin{proof}
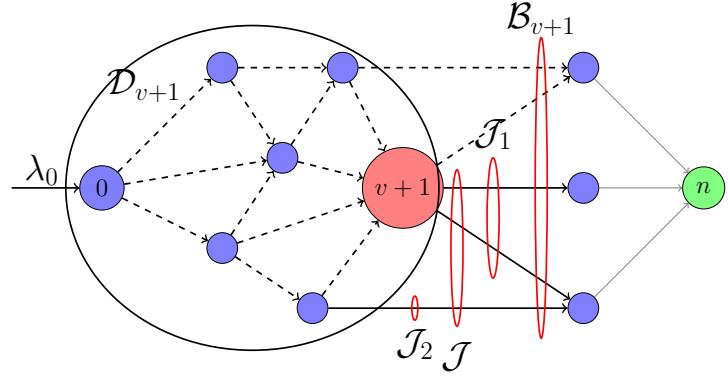
\begin{figure}
\begin{center}
\scalebox{0.8}
{
\begin{tikzpicture}

\path (0,0) node(a) [circle,draw,fill=blue!50!white] {$0$}
		(2,2) node (b) [circle,draw,fill=blue!50!white] {\text{  }}
	(4,2) node (c) [circle,draw,fill=blue!50!white] {\text{  }}
	(3,0.5) node (d) [circle,draw,fill=blue!50!white] {\text{  }} 
	(2,-1) node (e) [circle,draw,fill=blue!50!white] {\text{  }} 	
	(3.5,-2) node (f) [circle,draw,fill=blue!50!white] {\text{  }} 
		(5, 0) node (g) [circle,draw,fill=red!50!white] {$v+1$}
	(8,2) node (h) [circle,draw,fill=blue!50!white] {\text{  }} 
	(8,0) node (i) [circle,draw,fill=blue!50!white] {\text{  }} 
	(8,-2) node (j) [circle,draw,fill=blue!50!white] {\text{  }} 	
	(10, 0) node (k) [circle,draw,fill=green!50!white] {$n$};
	
\draw[thick,->] (-1.5,0) -- (a);	
\draw[dashed,thick,->] (a) -- (b);
\draw[dashed,thick,->] (b) -- (c);
\draw[dashed,thick,->] (b) -- (d);
\draw[dashed,thick,->] (a) -- (e);
\draw[dashed,thick,->] (a) -- (d);	
\draw[dashed,thick,->] (d) -- (c);
\draw[dashed,thick,->] (e) -- (d);
\draw[dashed,thick,->] (e) -- (f);
\draw[dashed,thick,->] (f) -- (g);
\draw[dashed,thick,->] (e) -- (g);
\draw[dashed,thick,->] (d) -- (g);
\draw[dashed,thick,->] (c) -- (g);
\draw[dashed,thick,->] (c) -- (h);
\draw[dashed,thick,->] (g) -- (h);
\draw[thick,->] (g) -- (i);
\draw[thick,->] (g) -- (j);
\draw[thick,->] (f) -- (j);
\draw[gray,->] (h) -- (k);
\draw[gray,->] (i) -- (k);
\draw[gray,->] (j) -- (k);	
	
\draw[thick] (2.5,0) circle (3.1 cm and 2.7 cm);
\draw[thick,red] (7.3,0) circle (0.1 cm and 2.5 cm);
\draw[thick,red] (5.9,-1) ellipse (0.1 cm and 1.3 cm);
\draw[thick,red] (6.5,-0.5) ellipse (0.1 cm and 1 cm);
\draw[thick,red] (5.2,-2) ellipse (0.05 cm and 0.2 cm);

\draw (-1,0.3) node {\Large $\lambda_0$};
\draw (0.7,1.7) node {\Large $\mathcal{D}_{v+1}$};	
\draw (7.3,2.8) node {\Large $\mathcal{B}_{v+1}$};	
\draw (5.9,-2.8) node {\Large $\mathcal{J}$};	
\draw (6.5,0.9) node {\Large $\mathcal{J}_1$};	
\draw (5.2,-2.6) node {\Large $\mathcal{J}_2$};

\end{tikzpicture}
}
\caption{Illustration of the sets used in proving the induction step.}
\label{fig:diffusivity}
\end{center}
\end{figure}
First recall that Theorem \ref{thm:uniquelimitflow} can be applied to the perturbed dynamical network (\ref{pertdynsyst}) in order to prove existence of a globally attractive limit flow $\tilde f^*\in\cl(\mc F)$ for the perturbed dynamical network flow (\ref{pertdynsyst}). For brevity in notation, for every $1\le v< n$, put 
$$\lambda_v^{*}:=\sum_{e \in \mc E_v^+} \floweqe\,,\qquad\tilde\lambda_v^{*}:=\sum_{e \in \mc E_v^-} \tilde f_e\,,\qquad\lambda_v^{\max}:=\sum_{e\in\mc E^+_v}\tilde f_e^{\max}\,.$$
Also, for every node $v\in\mc V$, let
$$\mc D_v := \bigcup\nolimits_{u=0}^v \edgeset_{u}^+\,,\qquad\mc B_v:=\{(u,w)\in\mc E:\,0\le u\le v,\,v<w\le n\}$$
be, respectively, the set of all outgoing links, and the link-boundary of the node set $\{0,1,\ldots,v\}$. 
 
We shall prove the following through induction on $u=0,1,\ldots,n-1$:
\be
\label{eq:ind-stat}
\sum_{e \in \mc J} \left( \tilde f_e^* - \floweqe \right) \leq \sum_{e \in \mc D_u} \delta_e\,,\qquad \forall\mc J\subseteq\mc B_u\,.
\ee

First, notice that $\mc B_0=\mc D_0=\edgeset_{0}^+$. Since $$\sum_{e \in \mc E_0^+} \delta_e \leq \delta < R(\mc N,f^*) \leq \sum_{e \in \mc E_0^+} (\flowmaxe - \floweqe)\,,$$ we also have that $\lambda_0 < \tilde\lambda^{\max}_v$. Therefore, by using \eqref{diffusivity} of Lemma~\ref{lemmadiffusivity}, one can verify that \eqref{eq:ind-stat} holds true for $v=0$. 

Now, for some $v \le n-2$, assume that \eqref{eq:ind-stat} holds true for every $u \le v$. Consider a subset $\mc J \subseteq \mc B_{v+1}$ and let $\mc J_1:= \mc J \cap \edgeset_{v+1}^+$ and $\mc J_2 :=\mc J \setminus\mc  J_1$ (e.g., see Figure~\ref{fig:diffusivity}). 
By applying Lemma~\ref{lemmadiffusivity} to the set $\mc J_1$, one gets that
\begin{equation}\label{eq:lemma-application}
\sum\nolimits_{e \in \mc J_1} \left(\tilde f_e^*-\floweq_e \right) \le \left[\tilde\lambda^{*}_{v+1}-\lambda^{*}_{v+1}\right]_+ + \sum\nolimits_{e \in \edgeset_{v+1}^+} \delta_e, \quad \forall \, t \geq 0.
\end{equation}
It is easy to check that $\mc J_2 \subseteq\mc  B_v$ and $\mc E_{v+1}^- \subseteq \mc B_v$.
Therefore, using \eqref{eq:ind-stat} for the sets $\mc J_2$ and $\mc J_2 \union \edgeset_{v+1}^-$, one gets the following inequalities respectively: 
\begin{align}
\label{eq:induct-case1}
\sum\nolimits_{e \in\mc J_2} \left(\tilde f_e^*-\floweqe \right)  \leq & \sum\nolimits_{e \in \mc D_v} \delta_e, \\
\label{eq:induct-case2}
\sum\nolimits_{e \in\mc J_2} \left(\tilde f_e^*-\floweqe \right) + 
\sum\nolimits_{e \in \edgeset_{{v+1}}^-} \left(\tilde f_e^*-\floweqe \right)  \leq & 
\sum\nolimits_{e \in \mc D_v} \delta_e.
\end{align}
Consider the two cases: $\tilde\lambda^{*}_{v+1} \leq\lambda^{*}_{v+1}$, or $\tilde\lambda^{*}_{v+1} > \lambda^{*}_{v+1}$. 
By adding up \eqref{eq:lemma-application} and \eqref{eq:induct-case1}, in the first case, or \eqref{eq:lemma-application} and \eqref{eq:induct-case2} in the second case, one gets that 
$$
 \sum_{e \in\mc J} \left( \tilde f^*_e-\floweq_e \right)
=\sum_{e \in\mc J_1} \left( \tilde f^*_e-\floweq_e \right) + \sum_{e \in\mc J_2} \left( \tilde f_e^*-\floweq_e \right)
\leq\sum_{e \in \edgeset_{v+1}^+} \delta_e +  \sum_{e \in \mc D_v} \delta_e \leq\sum_{e \in \mc D_{v+1}} \delta_e\,.$$
This proves \eqref{eq:ind-stat} for node $v+1$ and hence the induction step. 

Fix $1\le v<n$. Since $\mc E_{v}^- \subseteq \mc B_{v-1}$, \eqref{eq:ind-stat} with $u=v-1$ implies that 
$$\tilde\lambda_{v}^*=\sum_{e \in \mc E_{v}^-} \tilde f^*_e 
\leq\sum_{e \in \mc E_{v}^-} \floweqe + \sum_{e \in \mc D_{v-1}} \delta_e 
=\sum_{e \in \mc E_{v}^+} \floweqe + \sum_{e \in \mc E} \delta_e - \sum_{e \in \mc E \setminus \mc D_{v-1}} \delta_e\,,$$
where the third step follows from the fact that $\sum_{e \in \mc E_{v}^-} \floweqe = \sum_{e \in \mc E_{v}^+} \floweqe$ by conservation of mass. 
Then, since $\mc E_{v}^+ \subseteq \mc E \setminus \mc D_{v-1}$, one gets that
$$
\ba{rcl}\tilde\lambda_{v}^*
&\leq &\ds \sum\nolimits_{e} \floweqe + \delta - \sum\nolimits_{e} \delta_e  \\[7pt]
&< & \ds\sum\nolimits_{e} \floweqe + R(\mc N,f^*) - \sum\nolimits_{e} \delta_e \\[7pt]
&\leq & \ds\sum\nolimits_{e} \floweqe + \sum\nolimits_{e} \left(\flowmaxe - \floweqe \right) - \sum\nolimits_{e} \delta_e \\[7pt]
&=&\ds \sum\nolimits_{e} \left(\flowmaxe - \delta_e \right) \\[7pt]
&=&\ds \sum\nolimits_{e} \newflowmaxe\,,
\ea$$
where the summation index $e$ runs over $\mc E^+_v$. Hence, it follows from \cite[Lemma 2]{PartI} applied to the perturbed local system (\ref{pertlocalsys}) that 
\be\label{xlemma4.2}\tilde f_e^*=\tilde f_e^*(\tilde\lambda_{v}^*)<\tilde f_e^{\max}\,,\qquad \forall e\in\mc E^+_{v}\,,\ee
for all $1\le v< n-1$. Moreover, since $\lambda_0=\sum_{e\in\mc E^+_v}f_e^*<\sum_{e\in\mc E^+_v}f_e^{\max}\,,$ applying \cite[Lemma 2]{PartI} again to the perturbed local system (\ref{pertlocalsys}) shows that (\ref{xlemma4.2}) holds true for $v=0$ as well. Hence, $$\tilde f_e^*<f_e^{\max}\,,\qquad \forall e\in\mc E\,,$$ so that the limit flow  $\tilde f^*$ belongs to $\mc F$, and hence it is necessarily an equilibrium flow of the perturbed dynamical flow network (\ref{pertdynsyst}), as argued in Remark \ref{remark:limitflowequilibrium}. Therefore, the dynamical flow network (\ref{pertdynsyst}) is fully transferring.  
\end{proof}\medskip

Theorem \ref{maintheo} now immediately follows from Lemma \ref{lemmainduction}, and the arbitrariness of the admissible perturbation of magnitude smaller than $R(\mc N,f^*)$.

 \bibliographystyle{ieeetr}%
  \bibliography{KS-transportation}

\begin{thebibliography}{10}

\bibitem{Como.Savla.ea:MTNS10}
G.~Como, K.~Savla, D.~Acemoglu, M.~A. Dahleh, and E.~Frazzoli, ``On robustness
  analysis of large-scale transportation networks,'' in {\em Proc. of the Int.
  Symp. on Mathematical Theory of Networks and Systems}, pp.~2399--2406, 2010.

\bibitem{Simonsen.Buzna.ea:08}
I.~Simonsen, L.~Buzna, K.~Peters, S.~Bornholdt, and D.~Helbing, ``Transient
  dynamics increasing network vulnerability to cascading failures,'' {\em
  Physical Review Letters}, vol.~100, no.~21, pp.~218701--1 -- 218701--4, 2008.

\bibitem{PartI}
G.~Como, K.~Savla, D.~Acemoglu, M.~A. Dahleh, and E.~Frazzoli, ``Robust
  distributed routing in dynamical flow networks. {P}art {I}: locally
  responsive policies and weak resilience,'' {\em IEEE Transactions on
  Automatic Control}, 2011.
\newblock Submitted.

\bibitem{Como.Savla.ea:Wardrop-arxiv}
G.~Como, K.~Savla, D.~Acemoglu, M.~A. Dahleh, and E.~Frazzoli, ``Stability
  analysis of transportation networks with multiscale driver decisions,'' {\em
  SIAM Journal on Control and Optimization}, 2011.
\newblock Submitted, Available at \texttt{http://arxiv.org/abs/1101.2220}.

\bibitem{Sengoku.Shinoda.ea:88}
M.~Sengoku, S.~Shinoda, and R.~Yatsuboshi, ``On a function for the
  vulnerability of a directed flow network,'' {\em Networks}, vol.~18, no.~1,
  pp.~73--83, 1988.

\bibitem{Tassiulas.Ephremides:92}
L.~Tassiulas and A.~Ephremides, ``Stability properties of constrained queueing
  systems and scheduling policies for maximum throughput in multihop radio
  networks,'' {\em IEEE Transactions on Automatic Control}, vol.~37, no.~12,
  pp.~1936--1948, 1992.

\bibitem{Low.Paganini.ea:02}
S.~H. Low, F.~Paganini, and J.~C. Doyle, ``Internet congestion control,'' {\em
  IEEE Control Systems Magazine}, vol.~22, no.~1, pp.~28--43, 2002.

\bibitem{Dubey:86}
P.~Dubey, ``Inefficiency of {N}ash equilibria,'' {\em Mathematics of Operations
  Research}, vol.~11, no.~1, pp.~1--8, 1986.

\bibitem{Roughgarden:05}
T.~Roughgarden, {\em Selfish Routing and the Price of Anarchy}.
\newblock MIT Press, 2005.

\bibitem{Motter.Lai:02}
A.~E. Motter and Y.~Lai, ``Cascade-based attacks on complex networks,'' {\em
  Physical Review {E}}, vol.~66, no.~6, pp.~065102--1--065102--4, 2002.

\bibitem{Crucitti.Latora.ea:04}
P.~Crucitti, V.~Latora, and M.~Marchiori, ``Model for cascading failures in
  complex networks,'' {\em Physical Review E}, vol.~69, no.~4,
  pp.~045104--1--045104--4, 2004.

\bibitem{Cormen.Leiserson:90}
T.~H. Cormen, C.~E. Leiserson, R.~L. Rivest, and C.~Stein, {\em Introduction to
  Algorithms}.
\newblock MIT Press, 2nd~ed., 2001.

\bibitem{Ahuja.Magnanti.ea:93}
R.~K. Ahuja, T.~L. Magnanti, and J.~B. Orlin, {\em Network Flows: Theory,
  Algorithms, and Applications}.
\newblock Prentice Hall, 1993.

\bibitem{Borkar.Kumar:03}
V.~S. Borkar and P.~R. Kumar, ``Dynamic {C}esaro-{W}ardrop equilibration in
  networks,'' {\em IEEE Transactions on Automatic Control}, vol.~48, no.~3,
  pp.~382--396, 2003.

\bibitem{Bertsekas:99}
D.~Bertsekas, {\em Nonlinear Programming}.
\newblock Athena Scientific, 2~ed., 1999.

\bibitem{Wardrop:52}
J.~G. Wardrop, ``Some theoretical aspects of road traffic research,'' {\em ICE
  Proceedings: Engineering Divisions}, vol.~1, no.~3, pp.~325--362, 1952.

\bibitem{Beckmann.McGuire.ea:56}
M.~Beckmann, C.~B. McGuire, and C.~B. Winsten, {\em Studies in the Economics of
  Transportation}.
\newblock Yale University Press, 1956.

\bibitem{Patriksson:94}
M.~Patriksson, {\em The Traffic Assignment Problem: Models and Methods}.
\newblock V.S.P. Intl Science, 1994.

\bibitem{CVX}
M.~Grant and S.~Boyd, ``{CVX}: Matlab software for disciplined convex
  programming, version 1.21.'' \texttt{http://cvxr.com/cvx}, Feb. 2011.

\bibitem{Garavello.Piccoli:06}
M.~Garavello and B.~Piccoli, {\em Traffic Flow on Networks}.
\newblock American Institute of Mathematical Sciences, 2006.

\end{thebibliography}

\end{document}